\documentclass[11pt,a4paper]{amsart}
\usepackage[utf8]{inputenc}
\usepackage[english]{babel}
\usepackage[T1]{fontenc}
\usepackage{amsmath}
\usepackage{amsthm}
\usepackage{amsfonts}
\usepackage{amssymb}
\usepackage{lmodern}
\usepackage{mathrsfs}
\usepackage{physics}
\usepackage[dvipsnames]{xcolor}
\usepackage[colorlinks=true,citecolor=blue,linkcolor=BrickRed]{hyperref}
\usepackage{graphicx}
\usepackage[left=2.5cm,right=2.5cm,top=2.5cm,bottom=2.5cm]{geometry}
\usepackage{enumerate}
\usepackage{tikz-cd}

\usepackage{bm}
\usetikzlibrary{backgrounds}
\usetikzlibrary{calc}
\usetikzlibrary{hobby}
\usetikzlibrary{decorations.markings}
\usetikzlibrary{arrows.meta}
\usetikzlibrary{patterns}
\usepackage{caption}
\usepackage{subcaption}
\usepackage{array}
\usepackage{float}
\usepackage{afterpage}
\usepackage{hhline}
\usepackage{ytableau}

\numberwithin{equation}{section}

\DeclareMathOperator{\Ram}{Ram}
\DeclareMathOperator{\Irr}{Irr}
\DeclareMathOperator{\Slope}{Slope}

\DeclareMathOperator{\Id}{Id}

\DeclareMathOperator{\GL}{GL}

\newcommand{\cir}[1]{\langle #1 \rangle}

\begin{document}

\renewcommand{\proofname}{Proof}
\renewcommand{\Re}{\operatorname{Re}}
\renewcommand{\Im}{\operatorname{Im}}
\renewcommand{\labelitemi}{$\bullet$}

\newtheorem{theorem}{Theorem}[section]
\newtheorem{proposition}[theorem]{Proposition}
\newtheorem{lemma}[theorem]{Lemma}
\newtheorem{corollary}[theorem]{Corollary}
\newtheorem{conjecture}[theorem]{Conjecture}

\theoremstyle{definition}
\newtheorem{definition}[theorem]{Definition}
\newtheorem{notation}[theorem]{Notation}
\newtheorem{example}[theorem]{Example}
\newtheorem{remark}[theorem]{Remark}
\newtheorem*{claim}{Claim}


\title[Fourier transform of irregular connections and AD theories]{Fourier transform of irregular connections on $\mathbb P^1$ and classification of Argyres--Douglas theories}

\author{Jean Douçot}

\address[J.~Douçot]{`Simion Stoilow' Institute of Mathematics of the Romanian Academy,
	Calea Griviței 21,
	010702-Bucharest, 
	Sector 1, 
	Romania}
	\email{jeandoucot@gmail.com}

\maketitle

\begin{abstract}
We give a mathematical interpretation of the dualities between type $A$ Argyres--Douglas theories recently obtained by Beem, Martone, Sacchi, Singh and Stedman, building on work of Xie. Using the fact that, via the wild nonabelian Hodge correspondence, the data defining such a theory amount to singularity data for irregular connections on $\mathbb P^1$ of a specific form, we show that these dualities can all be realized as compositions of two types of more basic operations acting on such irregular connections: the Fourier transform and a Möbius transformation exchanging zero and infinity. The proof relies on the stationary phase formula giving explicit expressions for the singularity data of the Fourier transform. We also clarify the relation between the quivers describing the 3d mirrors of type $A$ Argyres--Douglas theories and the nonabelian Hodge diagrams defined in work of Boalch--Yamakawa and of the author: the 3d mirror corresponds to the unique nonabelian Hodge diagram with no negative edges/loops among those of singularity data in the corresponding orbit under basic operations.
\end{abstract}

\tableofcontents

\section{Introduction}

\subsection{Background and motivation} The goal of this work is to investigate the links between the classification problems for two different types of objects: moduli spaces of irregular connections on $\mathbb P^1$, and type $A$ Argyres--Douglas theories. 

\subsubsection{Moduli spaces in 2d gauge theory} 

Moduli spaces of irregular connections on complex algebraic curves are very rich mathematical objects. One aspect that makes them particularly interesting is that they are just one side (the de Rham side) of a larger picture. Indeed, as differentiable manifolds, they are isomorphic to moduli spaces for other types of objects defined on algebraic curves: meromorphic Higgs bundles (the Dolbeault side), via the irregular version of the nonabelian Hodge correspondence \cite{sabbah1999harmonic, biquard2004wild}; and generalized monodromy data, known as Stokes data (the Betti side) via the irregular Riemann--Hilbert correspondence. The Betti moduli spaces are called \emph{wild character varieties} since they generalize the usual character varieties, i.e. the moduli spaces of representations of the fundamental group of the underlying Riemann surface.  Furthermore, these \emph{nonabelian Hodge spaces} admit hyperkähler and symplectic structures, and they are phase spaces for (isospectral and isomonodromic) integrable systems, encompassing many systems of interest in mathematical physics, such as all Painlevé equations (see for instance \cite{boalch2018wild} for a survey.)

To obtain such a moduli space, one has to specify some initial data. They take the form of an \emph{irregular curve with boundary data}: this means a tuple $\bm\Sigma=(\Sigma, \mathbf a,\bm\Theta,\bm{\mathcal C})$, where $\Sigma$ is a smooth complex projective curve (the base curve on which connections are considered), $\mathbf a =\{a_1, \dots, a_n\}\subset \Sigma$ is a finite set of points in $\Sigma$ (the singularities of the connections), and the pair $(\bm\Theta, \bm{\mathcal C})$ encodes the types of the singularities of the connections at the points $a_i$. More precisely, on the Betti side, the \emph{irregular curve} $(\Sigma, \mathbf a,\bm\Theta)$ determines a wild character variety $\mathcal M_ B(\Sigma, \mathbf a,\bm\Theta)$, which admits an algebraic Poisson structure \cite{boalch2014geometry, boalch2015twisted}. Then any choice of \emph{boundary data}  determines a symplectic leaf $\mathcal M_B(\Sigma, \mathbf a,\bm\Theta,\bm{\mathcal C})$ of $\mathcal M_ B(\Sigma, \mathbf a,\bm\Theta)$. Here we will be concerned with the symplectic versions of wild character varieties. 

In genus zero, for $\Sigma=\mathbb P^1$, it turns out that there exist some highly non-trivial isomorphisms between moduli spaces corresponding to connections with different ranks, number of singularities and pole orders. At the level of the corresponding integrable systems, this is reflected by the fact that they admit several \emph{Lax representations}. This raises the question of the classification of such moduli spaces.\\

\subsubsection{Towards a classification in genus zero: diagrams and Fourier transform} A circle of ideas that has played an important role in making progress towards this problem is that moduli spaces of irregular connections on $\mathbb P^1$ are closely related to quiver varieties. An important motivation for this line of thought has been the work of Okamoto \cite{okamoto1992painleve}, who showed in the 90s that the Painlevé equations admit symmetry groups that are Weyl groups for some affine Dynkin diagrams. 

In the case of regular singularities, works of Nakajima and Crawley-Boevey have established that moduli spaces of connections with regular singularities on trivial bundles on $\mathbb P^1$ are isomorphic to quiver varieties associated to star-shaped quivers \cite{nakajima1994instantons, crawley2003matrices}. These results have then been extended by Boalch and Hiroe--Yamakawa to certain irregular cases, involving one (unramified) irregular singularity together with regular singularities \cite{boalch2008irregular, boalch2012simply, hiroe2014moduli}: the corresponding moduli spaces of irregular connections on trivial bundles on $\mathbb P^1$ are isomorphic to quiver varieties for so-called \emph{supernova quivers}. The moduli spaces corresponding to the Painlevé equations with number VI, V, IV, and II fit into this framework, and the quivers that appear are exactly the Dynkin diagrams found by Okamoto. In the \emph{simply-laced} case, i.e. when the irregular singularity is a pole of order at most 3, remarkably, this gives a clean graphical way to identify many isomorphisms between different moduli spaces of irregular connections: the same simply-laced supernova quiver can be obtained from connections with different ranks and number of singularities, and these different `readings' of the quiver correspond to isomorphisms between the moduli spaces \cite{boalch2012simply, boalch2016global}.

Most such isomorphisms are induced by a class of \emph{basic operations} on irregular connections on $\mathbb P^1$, including notably the \emph{Fourier transform} (or Laplace, or Fourier--Laplace transform). It acts in a complicated way, changing the number of singularities, the ranks, and the types of singularities (for instance it transforms a connection with regular singularities into one with an irregular singularity at infinity). Other types of basic operations include Möbius transformations (changes of global coordinate on $\mathbb P^1$), and twists (tensoring by a given rank one connection). In particular, it was shown by Deligne and Arinkin \cite{deligne2006letter, arinkin2010rigid}, extending work of Katz \cite{katz1996rigid}, that these operations allow one to obtain a complete classification of moduli spaces of irregular connections of $\mathbb P^1$ in the \emph{rigid} case, i.e. when the moduli space is reduced to a point. Some simple non-rigid irregular cases have been recently considered by the author \cite{doucot2026simplification}.

More recently, this picture has been partially generalized to connections on $\mathbb P^1$ with arbitrary singularities by Boalch--Yamakawa as well as the author \cite{boalch2020diagrams, doucot2021diagrams, doucot2024basic}: it is possible to associate a \emph{diagram} (i.e. a generalized graph, with possibly negative edge/loop multiplicities) to any irregular connection on $\mathbb P^1$, in such a way that the diagram is invariant under Fourier transform, so that we will have several readings of it, as in the simply-laced case, which are again expected to correspond to isomorphisms between the corresponding wild character varieties\footnote{In general finding explicit isomorphisms of wild character varieties induced by the Fourier transform is a difficult problem \cite{malgrange1991equations, mochizuki2010note, mochizuki2018stokes}. Some explicit isomorphisms are known in certain cases \cite{balser1981reduction,sabbah2016differential,boalch2016global,dagnolo2020topological, hohl2022d_modules,doucot2025topological}.}. However, the diagrams are not invariant under Möbius transformations, due to the fact that they depend on the choice of point at infinity on the Riemann sphere. This framework applies in particular to the moduli spaces associated to the remaining Painlevé equations.

\subsubsection{Classification of type $A$ Argyres--Douglas theories}

On the other hand, understanding the landscape of superconformal quantum field theories in 4 dimensions with $\mathcal N=2$ supersymmetry has been an active area of research in theoretical physics. Among those, several classes of theories of interest can be constructed from the initial data of a Riemann surface with punctures and some singularity data there, i.e. essentially from an irregular curve with boundary data in our language. 

A first class of such theories consists of theories of \emph{class $\mathcal S$}, introduced by Gaiotto, Moore, and Neitzke \cite{gaiotto2013wall}. They are defined by performing a twisted compactification of a 6d $\mathcal N=(2,0)$ superconformal field theory on a Riemann surface with punctures. The corresponding Hitchin system plays an important role in the study of these theories: the Seiberg--Witten curve corresponds to the spectral curve, the base of the Hitchin system is the Coulomb branch of the 4d theory, and the Dolbeault moduli space corresponds to the Coulomb branch of the theory compactified on a circle \cite{neitzke2015hitchin}. Another rich class of 4d $\mathcal N=2$ SCFTs is given by the Argyres--Douglas theories \cite{argyres1995new}. By definition, they are theories where Coulomb branch operators have fractional scaling dimensions. Many of them can be constructed by taking some limits of class $\mathcal S$ theories, and as a consequence can be defined from the datum of a Riemann surface with punctures of a particular form \cite{xie2013general}. They have been studied extensively in the last decade, see e.g. \cite{buican2015superconformal,wang2016classification, song2017vertex,creutzig2017w-algebras, benvenuti2017lagrangians,xie2021chiral}. 

An important aspect of the study of quantum field theories has been to identify dualities between various theories, and 4d $\mathcal N=2$ theories have been a fertile ground for this search. Specifically, in this work, we will be interested in a class of dualities for type $A$ Argyres--Douglas theories recently found by Beem, Martone, Sacchi, Singh, and Stedman \cite{beem2025simplifying}, providing a wide-ranging generalization of some previously known particular cases. For that class of theories, the initial data take the form of irregular Higgs bundles on $\mathbb P^1$, with an irregular (typically ramified) singularity at infinity, and a regular singularity at zero.

\subsubsection{Relating the classifications}

Given that the initial data for both Argyres--Douglas theories coming from class $\mathcal S$ and wild nonabelian Hodge spaces take the form  of a Riemann surface with some singularity data, it is not surprising that various aspects of their respective study turn out to be related. There are indeed many known instances where mathematical objects related to 2d gauge theory play a prominent role in the study of some 4d $\mathcal N=2$ theories, besides the Hitchin system. For example, the Painlevé equations, which correspond to some of the simplest nontrivial examples of nonabelian Hodge spaces, have been related to certain $\mathcal N=2$ theories, including Argyres--Douglas ones \cite{bonelli2017painleve,bonelli2025refined, grassi2019argyres}. 
The star-shaped quivers, the supernova quivers, as well as some of the more recent twisted wild nonabelian diagrams encoding certain moduli spaces of connections on $\mathbb P^1$, also have a physics interpretation, as the quivers describing the 3d mirrors of some class $\mathcal S$ and Argyres--Douglas theories \cite{benini2010mirrors,xie2013general, delzotto2015complete,xie20213d}. 

One can thus expect that the dictionary between the two pictures can be extended further: it is natural to conjecture that isomorphisms between moduli spaces of connections on $\mathbb P^1$, and in particular the Fourier transform, should be related to dualities between class $\mathcal S$ or Argyres--Douglas theories. Relations between the Fourier transform of irregular connections and dualities between some quantum field theories have already been investigated by Luu \cite{luu2015duality,luu2016fourier}, in the somewhat different context of minimal model 2d quantum gravity (see also the recent work by Alameddine, Marchal and Hayford \cite{alameddine2025painleve} on the relation between minimal models and the Painlevé I hierarchy.)

In the Argyres--Douglas context, one can readily observe that some known dualities match with the Fourier transform of irregular connections, under the natural identification of initial data given by the wild nonabelian Hodge correspondence: for instance the well-known duality between the $(A_{N-1}, A_{k-1})$ and $(A_{k-1}, A_{N-1})$  theories \cite{cecotti2010r-twisting} amounts to the same transformation of the parameters $k,N$ as the Fourier transform applied to connections on $\mathbb P^1$ with an irregular singularity at $\infty$ such that all Stokes circles have slope $\frac{k+N}{N}$. Similarly, the irregular realizations for class $\mathcal S$ theories discussed in \cite[§6.3.1]{xie2013general} correspond to a particular case of the so-called Harnad duality \cite{adams1990dual, harnad1994dual,yamakawa2011middle}, which is an avatar of the Fourier transform.

\subsection{Main results: Argyres--Douglas dualities from operations on connections}~

In this work, we show that this correspondence between isomorphisms on both sides holds in a much more general context: we can interpret all dualities between type $A$ Argyres--Douglas theories found in \cite{beem2025simplifying} in terms of compositions of basic operations on irregular connections on $\mathbb P^1$, namely of the Fourier transform, a Möbius transformation exchanging zero and infinity, and some twists. This gives a simple derivation of the dualities, and a further consistency check: in particular this should imply that the corresponding nonabelian Hodge spaces are isomorphic. From a purely mathematical point of view, our results essentially give a classification of wild nonabelian Hodge spaces of a specific form under basic operations, i.e. they can be seen as another case study of a non-rigid Katz--Deligne--Arinkin algorithm.

The proof relies on the stationary phase formula \cite{malgrange1991equations,fang2009calculation,sabbah2008explicit,graham2013calculation} which provides a way to determine explicitly the singularity data of the Fourier transform of an irregular connection on $\mathbb P^1$.

\subsubsection{Irregular curves with boundary data of (generalized) AD-$A$ type}

The first step is to translate the data defining a type $A$ Argyres--Douglas theory into singularity data for certain irregular connections on $\mathbb P^1$, i.e. in our language a genus zero irregular curve with boundary data $\bm\Sigma=(\mathbb P^1,\mathbf a, \bm\Theta, \bm{\mathcal C})$. 
Since in the framework of \cite{beem2025simplifying}, the Argyres--Douglas theories are defined in terms of some irregular Higgs bundles on $\mathbb P^1$, the corresponding connections are obtained by the wild nonabelian Hodge correspondence \cite{biquard2004wild}. We will say that the irregular curves with boundary data that arise in this way are of \emph{standard AD-$A$ type}. However, to account for the intermediate steps appearing when realizing the dualities of \cite{beem2025simplifying} as a composition of simple operations, we will have to consider slightly more general irregular curves with boundary data: we will say that those are of \emph{generalized AD-$A$ type}. 

The basic features of these data are as follows (see §\ref{sec:irreg_curves} for the detailed definitions).
Recall that a global irregular class $\bm\Theta$ is a collection of \emph{Stokes circles} $\cir{q_i}$ encoding the exponential terms $e^{q_i}$ appearing in the horizontal sections of the connection around its singularities, counted with multiplicities, and the boundary data $\bm{\mathcal C}$ consists of a collection of conjugacy classes $\mathcal C_{\cir{q_i}}$ encoding the formal monodromies of the connections.

An irregular curve with boundary data $\bm\Sigma$ of generalized AD-$A$ type has underlying curve $\mathbb P^1$, a regular singularity at zero and an irregular singularity at $\infty$. At the irregular singularity, the wild Stokes circles all have the same \emph{slope} $k$. In the general case, the local irregular class at infinity also has a regular part; the \emph{type I} case corresponds to the case where this regular part is not present. In particular, $\bm\Sigma$ depends on a quadruple of relevant parameters $\mathcal T=(m, k,\mathcal C_0, \mathcal C_\infty)$ that we call the \emph{AD-$A$ parameter} of $\bm\Sigma$. Here:
\begin{itemize}
\item $m\geq 1$ is an integer, the number of wild Stokes circles of $\bm\Sigma$.
\item $k\in \mathbb Q_{>0}$ is the \emph{slope} of each of the $m$ wild Stokes circles. If we write $k=\frac{s}{r}$, with $s, r$ coprime, $s$ is the irregularity of the Stokes circles, and $r$ their ramification order.
\item $\mathcal C_0\subset \mathrm{GL}_N(\mathbb C)$ is the conjugacy class of the (formal) monodromy at zero. The integer $N$ corresponds to the rank of the connections. 
\item $\mathcal C_\infty\subset \mathrm{GL}_{N-mr}(\mathbb C)$  is the conjugacy class of the regular part of the formal monodromy at infinity.
\end{itemize}
In the standard case, $\mathcal C_0$ is unipotent, corresponding to a Young diagram $[Y]$, and $\mathcal C_\infty$ is regular semisimple (of rank zero in the type I case), so the relevant parameters just consist of a triple $\mathscr T=(m, k,[Y])$, that we call a \emph{reduced AD-$A$ parameter}. The corresponding Argyres--Douglas theory is $D^b_p(\mathfrak{sl}_N, [Y])$, with $p=ms$, $b=mr$, $N=\mathrm{rk}(Y)$, in the notations of \cite{beem2025simplifying} (see Fig. \ref{fig:dictionary_notations} for a more complete dictionary between notations).

\subsubsection{Elementary AD-$A$ operations}

The second step is to observe that some simple compositions of basic operations preserve the set of irregular curves with boundary data of generalized AD-$A$ type.

Let $F$ be the Fourier transform and $M$ the Möbius transformation $z\mapsto 1/z$ exchanging $0$ and $\infty$ on $\mathbb P^1$. In the type I case these \emph{elementary AD-$A$ operations} are $F$ itself and the compositions $\widetilde F^+:=FM$, $\widetilde F^-:=MF$. In the general case, we have to consider variants $F_\alpha$,$\widetilde F^+_\alpha$, $\widetilde F^-_\alpha$, where $\alpha\in \mathbb C^*$, obtained by conjugating the type I ones by a \emph{Kummer twist} $T_\alpha$ (see §\ref{sec:dualities_general_case}). 
If $\bm\Sigma$ is an irregular curve with boundary data of generalized AD-$A$ type, depending on the slope $k$ not all these elementary operations transform $\bm\Sigma$ into an irregular curve with boundary data of generalized AD-$A$ type; when this is the case we say that the operation is \emph{allowed}.

Our first main result is an explicit description of the transformation of parameters induced by allowed elementary AD-$A$ operations. In type I, we have:

\begin{theorem}[Prop. \ref{prop:transformation_parameter_elementary_operation_type_I}]
Let $\bm\Sigma$ be an irregular curve with boundary data of generalized type I AD-$A$ type, with parameter $\mathcal T=(m,k,[Y^0], [Y^\infty])$. Then, if $O$ is an allowed elementary AD-$A$-operation on $\mathcal T$, the irregular curve with boundary data $O\cdot \bm\Sigma$ is of generalized AD-$A$ type. 

Furthermore, in each case, we have the following explicit expression for the parameter $O\cdot \mathcal T$ of $O\cdot\bm\Sigma$, where $[\tau(Y)]$ denotes the Young diagram obtained from $[Y]$ by deleting its first column, and $h_1(Y)$ denotes the height of the first column of $[Y]$:
\begin{enumerate}
\item If $k>1$, then the irregular curve with boundary data $F\cdot \bm\Sigma$ is  of generalized type I AD-$A$ type, with parameter
\begin{equation}
F\cdot \mathcal T=\left(m,\frac{s}{s-r},[ms-h_1(Y^0),Y^{\infty}], [\tau(Y^{0})]\right) 
\end{equation}
\item The irregular curve with boundary data $\widetilde F^+\cdot \bm\Sigma$ is of generalized type I AD-$A$ type, with parameter
\begin{equation}
\widetilde F^+\cdot\mathcal T=\left(m,\frac{s}{s+r},[ms-h_1(Y^\infty),Y^{0}], \tau(Y^{\infty})\right) 
\end{equation}
\item If $k<1$, the irregular curve with boundary data $\widetilde F^-\cdot \bm\Sigma$ is of  generalized type I AD-$A$ type, with parameter
\begin{equation}
\widetilde F^-\cdot\mathcal T=\left(m,\frac{s}{r-s},[\tau(Y^{0})], [ms-h_1(Y^0),Y^{\infty}]\right) 
\end{equation}
\end{enumerate}
\end{theorem}
There is a similar statement in the general case, see Prop. \ref{prop:transf_parameter_elementary_transf_twist}.

In brief, elementary AD-$A$ transformations preserve the irregularity $s$, but change the ramification order $r$ of the wild Stokes circles at infinity, and they change the Young diagrams $[Y^0]$ and $[Y^\infty]$ by removing the first column of one of them (which one depends on the case), and transferring the `complement' to $ms$ of that column to the other diagram, cf Fig. \ref{fig:elementary_operation_young_diagram}.

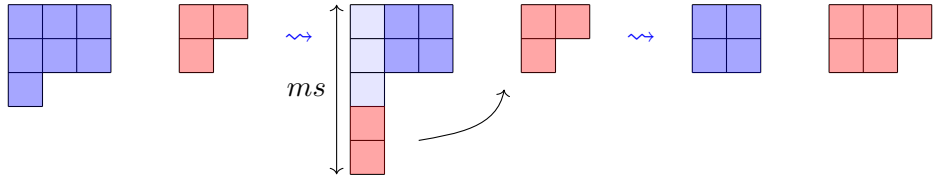
\begin{figure}[h]
\centering
\begin{tikzpicture}[scale=0.45]
\begin{scope}
\draw (0,0)--(3,0);
\draw (0,-1)--(3,-1);
\draw (0,-2)--(3,-2);
\draw (0,-3)--(1,-3);

\draw (0,0)--(0,-3);
\draw (1,0)--(1,-3);
\draw (2,0)--(2,-2);
\draw (3,0)--(3,-2);

\fill[blue, opacity=0.35] (0,0)--(3,0)--(3,-2)--(1,-2)--(1,-3)--(0,-3)--cycle;

\draw[blue] (8.5, -1) node {$\rightsquigarrow$};
\end{scope}

\begin{scope}[xshift=5cm]
\draw (0,0)--(2,0);
\draw (0,-1)--(2,-1);
\draw (0,-2)--(1,-2);

\draw (0,0)--(0,-2);
\draw (1,0)--(1,-2);
\draw (2,0)--(2,-1);

\fill[red, opacity=0.35] (0,0)--(2,0)--(2,-1)--(1,-1)--(1,-2)--(0,-2)--cycle;

\end{scope}

\begin{scope}[xshift=10cm]
\draw (0,0)--(3,0);
\draw (0,-1)--(3,-1);
\draw (0,-2)--(3,-2);
\draw (0,-3)--(1,-3);
\draw (0,-4)--(1,-4);
\draw (0,-5)--(1,-5);

\draw (0,0)--(0,-5);
\draw (1,0)--(1,-5);
\draw (2,0)--(2,-2);
\draw (3,0)--(3,-2);

\fill[blue, opacity=0.35] (1,0)--(3,0)--(3,-2)--(1,-2)--(1,-3)--cycle;

\fill[blue, opacity=0.1] (0,0)--(1,0)--(1,-3)--(0,-3)--cycle;

\fill[red, opacity=0.35] (0,-3)--(1,-3)--(1,-5)--(0,-5)--cycle;
\draw[->] (2,-4) to[out=10, in=-100] (4.5,-2.5);

\draw[<->] (-0.4,0) -- node[midway, left]{$ms$} (-0.4, -5);

\draw[blue] (8.5, -1) node {$\rightsquigarrow$};
\end{scope}

\begin{scope}[xshift=15cm]
\draw (0,0)--(2,0);
\draw (0,-1)--(2,-1);
\draw (0,-2)--(1,-2);

\draw (0,0)--(0,-2);
\draw (1,0)--(1,-2);
\draw (2,0)--(2,-1);

\fill[red, opacity=0.35] (0,0)--(2,0)--(2,-1)--(1,-1)--(1,-2)--(0,-2)--cycle;
\end{scope}

\begin{scope}[xshift=20cm]
\draw (0,0)--(2,0);
\draw (0,-1)--(2,-1);
\draw (0,-2)--(2,-2);

\draw (0,0)--(0,-2);
\draw (1,0)--(1,-2);
\draw (2,0)--(2,-2);

\fill[blue, opacity=0.35] (0,0)--(2,0)--(2,-2)--(0,-2)--cycle;
\end{scope}

\begin{scope}[xshift=24cm]
\draw (0,0)--(3,0);
\draw (0,-1)--(3,-1);
\draw (0,-2)--(2,-2);

\draw (0,0)--(0,-2);
\draw (1,0)--(1,-2);
\draw (2,0)--(2,-2);
\draw (3,0)--(3,-1);

\fill[red, opacity=0.35] (0,0)--(3,0)--(3,-1)--(2,-1)--(2,-2)--(0,-2)--cycle;
\end{scope}
\end{tikzpicture}
\caption{Effect of an elementary AD-$A$ operation on the Young diagrams of a type I  AD-$A$ parameter $\mathcal T=(m,k, [Y^0], [Y^\infty])$: it removes the column of one of the diagrams (represented in blue), and add its `complement' to the other diagram (represented). Which one of $[Y^0], [Y^\infty]$ is the blue/red diagram depends on the case.}
\label{fig:elementary_operation_young_diagram}
\end{figure}

This allows us to determine explicitly the structure of the `orbit' $\mathcal O(\mathcal T)$ of a generalized AD-$A$ parameter $\mathcal T$ under elementary operations, see Fig. \ref{fig:structure_orbit} and Fig. \ref{fig:structure_orbit_s=1}.

\subsubsection{Recovering the Argyres--Douglas dualities}

The three types of dualities found in \cite{beem2025simplifying}, given by equations (1.2), (1.3) and (1.4) in \emph{loc.~cit.}, can then be interpreted as compositions of elementary AD-$A$ operations:

\begin{theorem}[Prop. \ref{prop:duality_I}, \ref{prop:duality_II_complement_diagram}, \ref{prop:duality_III_gen_case}]
Let $\bm\Sigma$ be an irregular curve with boundary data of standard AD-$A$ type, with reduced parameter $\mathscr T=(m,k,[Y])$. Let us write $k=\frac{s}{r}>1$, with $s, r$ coprime.
\begin{itemize}
\item Assume that $\mathscr T$ is of type I. For any integer $l\geq 0$, $(\widetilde{F}^+)^l\cdot\bm\Sigma
$ is also of standard type I AD-$A$ type, with reduced parameter
\begin{equation}
(\widetilde{F}^+)^l\cdot\mathscr T:=\left(m, \frac{s}{ls+r}, [(ms)^l,Y]\right)
\end{equation}
\item Assume that $\mathscr T$ is of type I. Let $L$ be the number of columns of $[Y]$, and let us write $r=\kappa s +\rho$, with $1\leq \rho\leq s-1$. Assume that $s>1$ and $Ls>r$, i.e. $L>\kappa$.

Then the irregular curve with boundary data  $(\widetilde{F}^+)^{(L-1)}F\cdot\bm\Sigma
$ is also of type I AD-$A$ type, with reduced parameter
\begin{equation}
\label{eq:duality_II_intro}
(\widetilde{F}^+)^{L-1-\kappa} F (\widetilde{F}^-)^{\kappa} \cdot\mathscr T=\left(m, \frac{s}{Ls-r}, [Y^c]\right)
\end{equation}
where $[Y^c]$ is defined as follows: if $h_1\geq \dots\geq h_L$ are the heights of the columns of $[Y]$ from left to right, then $[Y^c]$ has $L$ columns with heights from left to right given by $(ms-h_L)\geq \dots \geq (ms-h_1)$.
\item Assume we are in the non type I case. Let $\mathcal T=(m, k, \mathcal C_0, \mathcal C_\infty)$, be a non-reduced parameter associated to $\mathscr T$. Let $\kappa=\mathrm{rk}(\mathcal C_0)-rm=\mathrm{rk}(\mathcal C_\infty)$, and $\beta_1, \dots, \beta_\kappa$ be the eigenvalues of $\mathcal C_\infty$. Then $\widetilde{F}^+_{\beta_\kappa}\dots\widetilde{F}^+_{\beta_1}\cdot \bm\Sigma$ has a parameter of the form
\begin{equation}
\widetilde{F}^+_{\beta_\kappa}\dots\widetilde{F}^+_{\beta_1}\cdot \mathcal T=\left(m, \frac{s}{\kappa s+r}, \mathcal C'_0, \{ \;\} \right)
\end{equation}
where $\{\;\}$ denotes the trivial (rank 0) conjugacy class, the nilpotent part of $\mathcal C'_0$ only depends on $\mathscr T$, and is given by the Young diagram
\begin{equation}
[(ms-1)^\kappa, Y].
\end{equation}
\end{itemize}
\end{theorem}

Indeed, via the dictionary of Fig. \ref{fig:dictionary_notations}, the transformations of the parameters match with the dualities of \emph{loc. cit}.

The situation can be understood nicely in terms of the structure of the orbit $\mathcal O(\mathcal T)$ given by Fig. \ref{fig:structure_orbit} (for $s>1$): the first duality amounts to following $l$ times a red arrow starting from $\mathcal T$, the second one amounts to following $L$ times a blue arrow, and the third one to following $\kappa$ times a red arrow (twisting by eigenvalues of the conjugacy class at infinity). In particular, for the second duality, we can understand explicitly the parameters corresponding to the intermediate steps, see Fig \ref{fig:duality_complement_intermediate_step}: the Young diagram $[Y]$ is transferred to the complement $[Y^c]$ one column at a time. Note that we can also interpret why extra hypermultiplets appear for the third duality in \emph{loc.~cit.}, see Remark \ref{remark:hypermultiplets}.

\subsubsection{3d mirrors as nonabelian Hodge diagrams} Finally, we clarify the relation between the quivers describing the 3d mirrors of type $A$ Argyres--Douglas theories and the wild nonabelian Hodge diagrams introduced in \cite{boalch2020diagrams, doucot2021diagrams} in the twisted case. These works introduce a construction which to any irregular connection $(E,\nabla)$ on $\mathbb P^1$ associates a diagram $\Gamma(E,\nabla)$, and (given a choice of marking), a dimension vector $\mathbf d$ for $\Gamma(E,\nabla)$. In the standard AD-$A$ case, the diagram and the dimension vector only depend on the reduced parameter $\mathscr T$.

We find that the 3d mirrors are indeed nonabelian Hodge diagrams, but there is an important subtlety: given a reduced AD-$A$  parameter $\mathscr T$, the 3d mirror is in general \emph{not} the diagram $\Gamma(\mathscr T)$. The problem is that  when $k<1$ the diagram $\Gamma(\mathscr T)$ is `bad', since it has negative edges and loops. However, if we consider the diagrams $\Gamma(\mathcal T')$ where $\mathscr T'$ is in the orbit $\mathcal O(\mathscr T)$ of $\mathscr T$ under elementary AD-$A$ operations, we can find a `good' diagram with nonnegative edges. Remarkably, this diagram is essentially unique, and corresponds precisely to the 3d mirrors found in the physics literature:

\begin{theorem}[Prop. \ref{prop:nonnegative_diagram_type_I}, \ref{prop:3d_mirrors_are_nonneg_diagrams_type_I}, \ref{prop:nonnegative_diagram_general_case}, \ref{prop:3d_mirrors_are_nonneg_diagrams_gen_case}] Let $\mathscr T$ be a reduced AD-$A$  parameter.  Then there exists a unique nonabelian diagram $\Gamma$ of the form $\Gamma(\mathcal T')$ for $\mathcal T'\in \mathcal O(\mathscr T)$ with no negative edges nor loops, and minimal number of edges, that we denote by $\Gamma_+(\mathscr T)$. Furthermore, the quiver describing the 3d mirror of the Argyres--Douglas theory defined by $\mathscr T$ is $\Gamma_+(\mathscr T)$ (with its uniquely defined dimension vector).
\end{theorem}

In terms of Fig. \ref{fig:structure_orbit}, the diagram is the one of the elements $\mathcal T^+$, $\mathcal T^-$ with slope $k'>1$ in the orbit $\mathcal O(\mathcal T)$, where $\mathcal T$ denotes the non-reduced parameter corresponding to $\mathscr T$. In particular, this sheds light on the fact that the 3d mirror quivers feature two legs although the initial data consist of a single nontrivial conjugacy class: the parameter $\mathcal T^+$ corresponds to an intermediate step of the duality  \eqref{eq:duality_II_intro} so part of the Young diagram of the conjugacy class has been transferred to the second leg (cf Fig. \ref{fig:duality_complement_intermediate_step}.)

\subsection*{Outlook} This work raises several questions about possible further relations between various aspects of the study of some 4d $\mathcal N=2$ theories, and moduli spaces of irregular connections on $\mathbb P^1$. One can for instance wonder whether some other dualities obtained in the physics literature might be interpreted in a similar way. Furthermore, in our approach we determine the nonnegative diagrams, i.e. the 3d mirror quivers, by using the Fourier transform and the stationary phase formula, while in the physics approach, when the regular puncture at zero is non maximal, they are obtained using some quiver algorithms such as the decay and fission algorithm \cite{bourget2024higgs,bourget2024decay}. It would thus be interesting to understand better the relation between the quiver algorithms and the Fourier transform. 

\subsection*{Structure of the article} The article is organized as follows. In section \ref{sec:irreg_curves}, we define the irregular curves with boundary of AD-$A$. In section \ref{sec:Fourier}, we discuss the Fourier transform and the stationary phase formula for irregular curves with boundary data of AD-$A$ type. In section \ref{sec:type_I_dualities}, we establish our main results in the type I case. In section \ref{sec:dualities_general_case} we introduce Kummer twists and deal with the general case. Finally in section \ref{sec:diagrams} we show that an irregular curve with boundary data of AD-$A$ type defines a unique minimal nonnegative nonabelian Hodge diagram, which coincides with the 3d mirrors of the corresponding type $A$ Argyres--Douglas theory.

\subsection*{Acknowledgements}  I thank P. Boalch and A. Bourget for pointing to me the article \cite{beem2025simplifying}. I am funded by the PNRR Grant CF 44/14.11.2022, ``Cohomological Hall Algebras of Smooth Surfaces and Applications''

\section{Irregular curves with boundary data and type $A$ Argyres--Douglas theories}
\label{sec:irreg_curves}

In this section, we recall the notion of irregular curves with boundary data (see e.g \cite{boalch2021topology, boalch2025twisted} for more details), and discuss the form of those corresponding to the data defining a type $A$ Argyres--Douglas theory. 

\subsection{Irregular curves with boundary data}

Let $\Sigma$ be a smooth complex projective curve (here we will only be interested in $\Sigma=\mathbb P^1$).

\subsubsection{Stokes circles}

Let $a\in \Sigma$. Let $\varphi_a:\widehat \Sigma_a\to \Sigma$ be the real oriented blow-up of $\Sigma$ at $a$. The preimage $\varphi_a^{-1}(a)=:\partial_a$ is a circle parametrizing directions around $a$. Let $z_a$ be a local coordinate on $\Sigma$ in a neighbourhood of $a$ vanishing at $a$. 

\begin{definition}
The (local) \emph{exponential local system} at $a$ is the local system of sets (i.e. the covering space) $\mathcal I_a\to \partial_a$ whose local sections are germs of functions on germs of sectors of the form
\begin{equation}
\label{eq:exponential_factor}
q=\sum_{k=1}^s b_k z_a^{-k/r},
\end{equation}
where $r\geq 1$ is an integer, and $b_k\in \mathbb C$ for $k=1, \dots, s$. 
\end{definition}

Such a section $q$ of $\mathcal I_a$ determines several numbers:
\begin{itemize}
\item 
The smallest possible integer $r\geq 1$ such that $q$ can be written in this way is the \emph{ramification order} $q$, denoted by $\Ram(q)$ (in particular for $q=0$ we have $\Ram(q)=1$). 
\item  The degree $s$ of $q$ as a polynomial in $z_a^{-1/r}$, with $r=\Ram(q)$, is the \emph{irregularity} of $q$, denoted by $\Irr(q)$.
\item The quotient $s/r\in \mathbb Q_{\geq 0}$ is then called the \emph{slope} of $q$, and denoted by $\Slope(q)$ (for $q=0$ we set $\Slope(q)=0$).
\end{itemize}

\begin{definition} A \emph{Stokes circle} at $a$ is a connected component of the topological space $\mathcal I_a$, i.e. an element of $\pi_0(\mathcal I_a)$.
\end{definition}

Given a local section $q$ of $\mathcal I_a$ as above, we denote by $\cir{q}$ its connected component in $\mathcal I_a$. It is homeomorphic to a circle, hence the terminology. The corresponding subcover $\cir{q}\to\partial_a$ has finite degree $r$ (this corresponds to the fact that there are $r$ determinations of the root $z_a^{1/r}$ in any direction). 

More concretely, the set of Stokes circles at $a$ is in bijection with the set of orbits of polynomials $\sum_{k=1}^s b_k z_a^{-k/r}$ in $z_a^{-1/r}$ without constant term, under the cyclic Galois action permuting the determinations of the $r$-th root of $z$, generated by $z_a^{-1/r}\mapsto e^{-2i\pi/r}z_a^{-1/r}$.

We call the Stokes circle $\cir{0}$ at $a$ the \emph{tame circle} at $a$, and refer to other ones as \emph{wild} Stokes circles.

\subsubsection{Local irregular classes and boundary data} An irregular class is a collection of Stokes circles with multiplicities:

\begin{definition}
Let $a\in \mathbb P^1$. A (local) irregular class at $a$ is a multiset of Stokes circles at $a$, i.e. a formal sum
\begin{equation}
\label{eq:irregular_class}
\Theta_a=\sum_{i=1}^{p} n_i \cir{q_i},
\end{equation}
where $\cir{q_i}\in \pi_0(\mathcal I_a)$ are pairwise distinct Stokes circles, and $n_i\in \mathbb Z_{\geq 1}$, for $i\in \{1, \dots, p\}$.

The \emph{rank} of the irregular class $\Theta_a$ is the integer
\begin{equation}
\mathrm{rk}(\Theta_a)=\sum_{i=1}^p n_i\Ram(q_i).
\end{equation}
\end{definition}

\begin{definition}
Let $\Theta_a=\sum_{i=1}^{p} n_i \cir{q_i}$ be an irregular class at $a$. A boundary datum $\mathcal C_a$ for $\Theta_a$ is the datum of a conjugacy class
\[
\mathcal C_{\cir{q_i}_a}\subset \mathrm{GL}_{n_i}(\mathbb C),
\]
for each $i\in \{1, \dots, p\}$.
\end{definition}

The well-known Turritin--Levelt theorem implies that any rank $N$ connection $\widehat \nabla_a$ on the formal punctured disc $\widehat D^{\times}_a:=\mathrm{Spec}(\mathbb C(\!(z_a)\!))$ canonically determines a pair $(\Theta_a, \mathcal C_a)$ where $\Theta_a$ is a rank $N$ irregular class at $a$ and $\mathcal C_a$  is a boundary datum for $\Theta_a$. The connection has a regular singularity at $a$ if $\Theta_a=N\cir{0}$, and an irregular singularity otherwise.

\subsubsection{Global case} 

\begin{definition} Let $N\geq 1$ be an integer.
\begin{itemize}
\item 
A rank $N$ irregular curve is a triple $(\Sigma, \mathbf a, \bm\Theta)$, where $\Sigma$ is a smooth complex projective curve, $\mathbf a =\{a_1, \dots a_n\}$ is a finite set of points on $\Sigma$, and $\bm\Theta=(\Theta_{a_1}, \dots, \Theta_{a_n})$ is the datum of an irregular class $\Theta_{a_i}$ of rank $N$ at each $a_i$. 
\item
A rank $N$ irregular curve with boundary data is a quadruple $(\Sigma, \mathbf a, \bm\Theta, \bm{\mathcal C})$, where $(\Sigma, \mathbf a, \bm\Theta)$ is an irregular curve, and $\bm{\mathcal C}=(\mathcal C_{a_1}, \dots, \mathcal C_{a_n})$, where $\mathcal C_{a_i}$ is a boundary datum for $\Theta_{a_i}$ for each $i\in \{1, \dots, n\}$.
\end{itemize}
\end{definition}

If $\Sigma$ is a complex projective curve,  $\mathbf a=\{a_1, \dots, a_n\}\subset \Sigma$ a finite set of points on $\Sigma$,  an algebraic connection $(E,\nabla)$ of rank $N$ on $\Sigma^{\circ}:=\Sigma\smallsetminus \mathbf a$ canonically determines a rank $N$ irregular curve, as well as boundary data, by considering its restriction to the formal punctured disk around each $a_i$ and taking the local formal data $(\Theta_{a_i}, \mathcal C_{a_i})$ at $a_i$.

By the work of Boalch and Boalch--Yamakawa \cite{boalch2014geometry, boalch2015twisted}, any irregular curve $(\Sigma, \mathbf a, \bm\Theta)$ determines a (possibly empty) wild character variety $\mathcal M_B(\Sigma, \mathbf a, \bm\Theta)$. Via the irregular Riemann--Hilbert correspondence, it parametrizes isomorphism classes of connections $(E,\nabla)$ on $\Sigma^\circ$ with irregular class $\Theta_{a_i}$ at $a_i$.

Furthermore, the variety $\mathcal M_B(\Sigma, \mathbf a, \bm\Theta)$ has an algebraic Poisson structure, and its symplectic leaves correspond to fixing boundary data $\bm{\mathcal C}$ for the irregular curve $(\Sigma, \mathbf a, \bm\Theta)$. We denote by $\mathcal M_B(\Sigma, \mathbf a, \bm\Theta, \bm{\mathcal C})$ the symplectic wild character varieties obtained in this way.

\begin{remark}
We say that an irregular curve with boundary data $\bm\Sigma$ is \emph{effective} if there exists an irreducible connection $(E,\nabla)$ with data $\bm\Sigma$. In the rest of the article, irregular curves with boundary data will always be assumed effective.
\end{remark}

\begin{notation}
For our purposes here, we will only consider the genus zero case, i.e. $\Sigma=\mathbb P^1$, and it will be convenient at times to use slightly different (but clearly equivalent) points of view on genus zero irregular curves with boundary data:
\begin{itemize}
\item We defined a global irregular class $\bm\Theta$ both as a collection $(\Theta_{a_1}, \dots, \Theta_{a_n})$ of irregular classes at different points on $\mathbb P^1$. We can also view it as a multiset of Stokes circles allowed to lie above different points of $\mathbb P^1$, i.e. a formal sum 
\begin{equation}
\bm\Theta=\sum_{i=1}^p n_i\cir{q_i},
\end{equation}
where $\cir{q_i}$ is a Stokes circle at some point $\pi(\cir{q_i})=a_i$ in $\mathbb P^1$, i.e. a connected component of the global exponential local system $\mathcal I:=\sqcup_{a\in \mathbb P^1} \mathcal I_a$ (for convenience we will sometimes write Stokes circles as $\cir{q}_a$, with a subscript indicating the corresponding singularity, sometimes without subscript.)
\item Similarly, we can view the pair $(\bm\Theta, \bm{\mathcal C})$ as a collection
\[
(\bm\Theta, \bm{\mathcal C})=\left\{ (\cir{q_i}, \mathcal C_{\cir{q_i}})\;\middle\vert\; i\in \{1, \dots, p\}\right\}
\]
of pairs $(\cir{q_i}, \mathcal C_{\cir{q_i}})$.
\end{itemize}
Furthermore, the singularity locus $\mathbf a$ is determined by the global irregular class $\bm\Theta$, so an irregular curve with boundary data is actually determined by the pair $(\bm\Theta, \bm{\mathcal C})$. By a slight abuse of language, we will thus sometimes view a genus zero irregular curves with boundary data as a pair $(\bm\Theta, \bm{\mathcal C})$, with $\cir{q_i}$ a Stokes circle, and $\mathcal C_{\cir{q_i}}$ a conjugacy class for $i\in\{1, \dots, p\}$.

A last point: associating boundary data $\bm{\mathcal C}$ to a connection $(E,\nabla)$ depends on a choice of orientation of $\mathbb P^1$ (if we reverse the orientation we get the inverse conjugacy classes). Here it will be convenient to adopt the following convention: for $a\neq \infty$ we use the standard orientation on $\mathbb P^1$ to define $\mathcal C_a$, but for $a=\infty$ we take the reverse orientation. This amounts to having the chosen orientation at each $a\in \mathbb P^1$ corresponding to the standard orientation on the complex plane, using the local coordinate $z_a$ defined by $z_a:=z-a$ for $a\neq \infty$, and $z_a:=z^{-1}$ for $a=\infty$ (where $z$ denotes the standard global coordinate on $\mathbb A^1$.) 
\end{notation}

\subsection{Type $A$ Argyres--Douglas case}

In the framework of \cite{beem2025simplifying}, the initial data required to define a type $A$ Argyres--Douglas theory consist of a Higgs bundle on $\mathbb P^1$ with an irregular pole at $\infty$, a regular pole at $0$, and prescribed form at the irregular singularity. Via the wild nonabelian Hodge correspondence of \cite{biquard2004wild}, this amounts on the de Rham side to connections whose irregular curves with boundary data are of a specific form.

\subsubsection{Most general form}

For our purposes here, it will be necessary to consider slightly more general irregular curves than those obtained in this way, for the following reason. We want to show that the dualities between type $A$ Argyres--Douglas theories derived in \cite{beem2025simplifying} correspond to a composition of simple operations on irregular connections on $\mathbb P^1$. The point is that when applying these operations successively, at the intermediate steps, the irregular curves with boundary data will not correspond any more to a Higgs bundle of the form considered in \emph{loc.~cit.} (concretely the conjugacy classes can become more general).

This motivates the following definition, which encompasses all irregular curves with boundary data that we will encounter:

\begin{definition}
\label{def:generalized_ADA_type}
Let $\bm{\Sigma}=(\mathbb P^1, \mathbf a, \bm \Theta, \bm{\mathcal C})$ be a genus zero irregular curve with boundary data, and $N=\mathrm{rk}(\bm\Sigma)$. We say that $\bm{\Sigma}$ is of generalized AD-$A$ type if
\begin{itemize}
 \item $\mathbf a=\{0, \infty\}$;
 \item the singularity at 0 is regular, i.e. $\Theta_0= N \cir{0}_0$.
 \item the irregular class at infinity is of the form
 \begin{equation}
  \Theta_\infty=\cir{q_1}_\infty+\dots+\cir{q_m}_\infty+ (N-mr)\cir{0}_\infty,
 \end{equation}
 where $\cir{q_1}_\infty, \dots, \cir{q_m}_\infty$ are pairwise distinct Stokes circles with the same ramification order $r>1$, same slope $k=\frac{s}{r}>0$ with $s,r$ coprime, such that the Galois orbits of their leading terms are pairwise distinct, i.e.  for $i\in\{1, \dots, m\}$ we have
\begin{equation}
q_i=\sum_{l=1}^{s} b_{il} z^{l/r},
\end{equation}
with $b_{il}\in \mathbb C$ for $l\in \{1, \dots, s\}$, $b_{is}\neq 0$, and $b_{is}^r\neq b_{js}^r$ for distinct $i,j\in \{1, \dots, m\}$.
\end{itemize}
\end{definition}

In brief, there are two singularities: an irregular one at $\infty$, a regular one at $0$, and we assume that the nonzero Stokes circles at the irregular singularity are all of multiplicity 1, and have the same slope $k$.

\begin{definition}
If $\bm\Sigma$ is as in Def. \ref{def:generalized_ADA_type}, we say that the tuple
\begin{equation}
\mathcal T:=(m,k, \mathcal C_{\cir{0}_0},\mathcal C_{\cir{0}_\infty})
\end{equation}
is the \emph{parameter} of $\bm\Sigma$.
\end{definition}

\begin{notation}
To simplify notations, since we will the conjugacy classes associated to the tame circles at $0$ and $\infty$ are the only ones playing an important role, we will write $\mathcal C_0$, $\mathcal C_\infty$ instead of $\mathcal C_{\cir{0}_0}$, $\mathcal C_{\cir{0}_\infty}$ respectively.
\end{notation}

\begin{notation}
Using the Jordan normal form, a conjugacy class $\mathcal C\subset\mathrm{GL}_n(\mathbb C)$ is determined by the datum of its (pairwise distinct) eigenvalues $\lambda_1, \dots, \lambda_p\in \mathbb C^*$ and of the Young diagram $Y_{\lambda_i}$ encoding the nilpotent part of its restriction to the characteristic subspace associated to $\lambda_i$ for each $i\in \{1, \dots, p\}$. We will then write
\[
\mathcal C=\{(\lambda_1, [Y_{\lambda_1}]), \dots, (\lambda_p, [Y_{\lambda_p}])\}.
\]
In particular, a unipotent conjugacy class $\mathcal C$ is of the form $\mathcal C=\{(1,[Y])\}$ for some Young diagram $Y$.

For the conjugacy classes $\mathcal C_0, \mathcal C_\infty$ in a generalized AD-$A$ parameter $\mathcal T$, we will write
\begin{align*}
\mathcal C_0&=\{(\alpha_1, [Y^0_{\alpha_1}]), \dots, (\alpha_k,[Y^0_{\alpha_k}])\},\\
\mathcal C_\infty&=\{(\beta_1, [Y^{\infty}_{\beta_1}]), \dots, (\beta_l,[Y^{\infty}_{\beta_l}])\},
\end{align*}
where $\alpha_1, \dots, \alpha_k\in \mathbb C^*$ are the eigenvalues of $\mathcal C_0$, and $Y^0_{\alpha_i}$, for $i\in\{1, \dots, k\}$, the corresponding Young diagrams, and similarly $\beta_1, \dots, \beta_l\in \mathbb C^*$ are the eigenvalues of $\mathcal C_\infty$, and $Y^0_{\beta_j}$, for $j\in\{1, \dots, l\}$, the corresponding Young diagrams.
\end{notation}

\subsubsection{Particular cases} We now introduce some particular subtypes of irregular curves with boundary data of generalized AD-$A$, type, obtained by requiring $\mathcal C_0$ and $\mathcal C_\infty$ to be of a specific form.

\begin{definition}
\label{def:particular_AD-A_types}
Let $\bm \Sigma$ be an irregular curve with boundary data of generalized AD-$A$ type, and $\mathcal T=(m, k, \mathcal C_0, \mathcal C_\infty)$ its parameter. Then we say that:
\begin{itemize}

\item $\bm \Sigma$ is of standard AD-$A$ type if:
\begin{itemize}
\item $\mathcal C_0$ is unipotent, and 
\item $\mathcal C_\infty$ is either of rank 0, or a regular semisimple conjugacy class, with pairwise distinct eigenvalues all different from 1. 
\end{itemize}

\item $\bm \Sigma$ is of standard type I AD-$A$ type if:
\begin{itemize}
\item $\mathcal C_0$ is unipotent, and 
\item $\mathcal C_\infty$ is of rank 0.
\end{itemize}

\item $\bm \Sigma$ is of generalized type I AD-$A$ type if:
\begin{itemize}
\item $\mathcal C_0$ is unipotent, and 
\item $\mathcal C_\infty$ is either of rank 0, or unipotent.
\end{itemize}
\end{itemize}
If $\bm\Sigma$ is of standard (possibly type I) AD-$A$ type, the conjugacy class $\mathcal C_0$ is of the form $\mathcal C_0=\{(1, [Y])\}$, for some Young diagram $Y$. We call the triple $\mathscr T:=(m,k, [Y])$ the \emph{reduced parameter} of $\bm\Sigma$.
\end{definition}

The motivation for these definitions is that the standard, and standard type I cases correspond via the wild nonabelian Hodge correspondence to the Higgs bundles providing the initial data for the type $A$ Argyres--Douglas theories considered in \cite{beem2025simplifying}, namely:
\begin{itemize}
\item A type $A$ Argyres--Douglas theory in the terminology of \emph{loc.~cit.} corresponds to an irregular curve with boundary data of standard AD-$A$ type.
\item A type $A$ Argyres--Douglas theory  of type I in the terminology of \emph{loc.~cit.} corresponds to an irregular curve with boundary data of standard type I AD-$A$ type.
\end{itemize}
(On the other hand, the generalized type I case will appear in the intermediate steps when realizing the AD dualities from successive application of simple operations on connections  in the type I case.)

More precisely, the dictionary between our notations and those of \cite{beem2025simplifying} for the parameters defining the AD theories is as follows: the AD theory denoted by $D_p^b(\mathfrak{sl}_N, [Y])$ corresponds to irregular curves of standard (possibly type I) type, with the reduced parameter $\mathscr T=(m, k, [Y])$ related to the parameters used in \emph{loc.~cit.} as indicated in the tables of Fig. \ref{fig:dictionary_notations} below.

\begin{figure}
\centering

\begin{subfigure}[b]{0.45\textwidth}
\centering

\begin{tabular}{|c|c|}
\hline
Irreg. curve  & AD theory\\
\hline
$m$ & $m=\mathrm{gcd}(p, b)$\\
\hline
$k$ & $\frac{n+k}{n}=\frac{p}{b}$\\
\hline
$r$ & $n$\\
\hline
$s$ & $n+k=q$\\
\hline
$N$ & $N$\\
\hline
$Y$ & $Y$\\
\hline
\end{tabular}
\end{subfigure}
\hfill
\begin{subfigure}[b]{0.45\textwidth}
\centering
\begin{tabular}{|c|c|}
\hline
AD theory & Irreg. curve\\
\hline
$m$ & $m$\\
\hline
$p$ & $ms$\\
\hline
$q$ & $s$\\
\hline
$b$ & $mr$\\
\hline
$N$ & $N=\rm{rk}(Y)$\\
\hline
$Y$ & $Y$\\
\hline
$k$ & $s-r$\\
\hline

\end{tabular}
\end{subfigure}
\caption{Dictionary between parameters for type $A$ Argyres--Douglas theories in the notations of \cite{beem2025simplifying} and the parameter of the irregular curve with boundary data of standard AD-$A$ type (possibly type I) in our notation here, in both directions. The unipotent conjugacy class $\mathcal C_0$ is given by $\mathcal C_0=\{(1, [Y])\}$.}
\label{fig:dictionary_notations}
\end{figure}

Let us briefly discuss how to see this from the form of the irregular Higgs bundles playing the role of initial data for Argyres--Douglas theories given in \cite[§2.1]{beem2025simplifying}.

To this end, let us first recall how to determine the irregular curve with boundary data of the connection determined by a given irregular Higgs bundle via the wild nonabelian Hodge correspondence of \cite{biquard2004wild} (actually the case that we need is more general than the untwisted one discussed in full detail in  \emph{loc. cit}, but things remain true in the twisted case as well, see also \cite{chuang2020twisted,eper2025rank}).

Consider a meromorphic Higgs bundle $(E, \Phi)$ on $\Sigma$, and let $a$ be a pole of the Higgs field $\Phi$. Let $z$ be a local coordinate on $\Sigma$ at $a$. We assume that $(E, \Phi)$ is `good' in the sense of \cite[Remark 6]{boalch2018wild}, i.e. locally, after passing to a cyclic cover i.e. setting $w:=z^{1/r}$ for some integer $r\geq 1$, in some trivialization the Higgs field $\Phi$ has a local model of the form
\begin{equation}
\label{eq:local_model_Higgs_bundle}
\left(\frac{T_{s+1}}{w^{s+1}}+\dots +\frac{T_2}{w^2}+\frac{T_1}{w}\right)dw+ \text{holomorphic terms}
\end{equation}
where $T_1, \dots, T_{s+1}$ are constant $N$ by $N$ matrices, with $T_2, \dots, T_{s+1}$ diagonal. 

By the wild nonabelian Hodge correspondence, $(E,\Phi)$ defines a (parabolic) connection $(E', \nabla)$, admitting a local model of the form
\begin{equation}
\label{eq:local_model_connection}
\left(\frac{A_{s+1}}{w^{s+1}}+\dots +\frac{A_2}{w^2}+\frac{A_1}{w}\right)dw + \text{holomorphic terms},
\end{equation}
where $A_1, \dots, A_{s+1}$ are constant $N$ by $N$ matrices, with $A_2, \dots, A_{s+1}$ diagonal,
where the matrices $A_i$ and $T_i$ are related as follows (from Thm. 0.1 and the remarks on the nilpotent parts in \cite{biquard2004wild}):
\begin{itemize}
\item $A_i=2 T_i$ for $i\geq 2$. 
\item The eigenvalues of $A_1$, $T_1$ and the parabolic weights on the connection and Higgs bundle sides are related by Simpson's table.
\item The nilpotent parts of $A_1$ and $T_1$ are the same.
\end{itemize}
Now, the corresponding irregular class $\Theta_a$ is obtained as follows: define the \emph{irregular type} $Q$ of the local model by 
\begin{equation}
\label{eq:def_irreguar_type}
dQ=\left(\frac{A_{s+1}}{w^{s+1}}+\dots +\frac{A_2}{w^2}\right)dw,
\end{equation}
with
\[
Q=\begin{pmatrix}
\widehat{q}_1 & & \\
& \ddots &\\
& & \widehat{q}_N
\end{pmatrix}
\]
where $\widehat q_i\in w^{-1}\mathbb C[w^{-1}]$ for $i\in \{1, \dots, N\}$. Here the \emph{exponential factors} $\widehat q_i$ are not necessarily distinct; up to reordering the local basis we can write without loss of generality
\[
Q=\mathrm{diag}(\underbrace{\widetilde q_1, \dots, \widetilde q_1}_{\widetilde n_1 \text{ times}}, \dots, \underbrace{\widetilde q_k, \dots, \widetilde q_k}_{\widetilde n_k \text{ times}})
\]
with $\widetilde q_i\in w^{-1}\mathbb C[w^{-1}]$ for $i\in \{1, \dots, k\}$
with the $\widetilde q_i$ for $i\in \{1, \dots, k\}$ now pairwise distinct. 

Now, since $(E', \nabla)$ is meromorphic in $z$, if a Puiseux polar part $q\in z^{-1/r}\mathbb C[z^{-1/r}]$ is present in $Q$ with multiplicity $n$, then so are all its Galois conjugates (with the same multiplicity). Denoting by $\cir{q_1}_a, \dots,\cir{q_p}_a$ the Stokes circles corresponding to the distinct Galois orbits appearing in $Q$, and $n_i$, for $i\in \{1, \dots, p\}$ the corresponding multiplicities, we obtain a well-defined irregular class
\[
\Theta_a=\sum_{i=1}^p n_i\cir{q_i}_a,
\]

Finally, the boundary datum $\mathcal C_a$ for $\Theta_a$ is determined by the residue $A_1$, as follows. Using meromorphic gauge transformations, we can assume that $A_1$ is in the centralizer of $Q$, i.e. is a block diagonal matrix
\begin{equation}
\label{eq:additive_formal_monodromy}
A_1=\begin{pmatrix}
\Lambda_1 & & \\
& \ddots & \\
& & \Lambda_k
\end{pmatrix} 
\end{equation}
with $\Lambda_i$ a constant square matrix of size $\widetilde n_j$, for $j\in\{1, \dots, k\}$.
Then for $i\in \{1, \dots, p\}$ the conjugacy class $\mathcal C_{\cir{q_i}_a}$ associated to $\cir{q_i}_a$ is the conjugacy class of $\exp(2\sqrt{-1}\pi\Ram(q_j)\Lambda_j)$ for any $j\in \{1, \dots, k\}$ such that $\cir{\widetilde q_j}_a=\cir{q_i}_a$.

One can then apply this to the Higgs bundles described in \cite[§2.1]{beem2025simplifying} corresponding to the Argyres--Douglas theory denoted by $D_p^b(\mathfrak{sl}_N, [Y])$ there (choosing vanishing parabolic weights for the regular singularity at zero, and generic weights for the irregular singularity at infinity).

One obtains that the corresponding irregular curve with boundary data is of standard AD-$A$ type as per Def. \ref{def:particular_AD-A_types}, with parameter $\mathcal T=(m,k,\mathcal C_0, \mathcal C_\infty)$ such that $m, k, \mathcal C_0$ are given in terms of $b,p, N$ and $[Y]$ by the tables in Fig. \ref{fig:dictionary_notations}, and $\mathcal C_\infty$ is either of rank 0 (in the type I case), or generic regular semisimple, with eigenvalues different from 1. Indeed, for the regular singularity at zero $\mathcal C_0$, the residue $\Lambda$ of the Higgs field is nilpotent, given by the Young diagram $[Y]$, so the conjugacy class $\mathcal C_0$ is the conjugacy class of $\exp(2\sqrt{-1}\pi \Lambda)$, hence $\mathcal C_0$ is unipotent, and $\mathcal C_0=\{(1,[Y])\}$.

\begin{example} Let us discuss a few simple examples:

\begin{enumerate}

\item Let $r, s\geq 1$ be two coprime integers, with $s>r$. The irregular curve 
\[
(\mathbb P^1, \infty, \cir{z^{s/r}}_\infty)
\]
corresponds to the Argyres--Douglas theory often referred to as the $(A_{s-1}, A_{s-r-1})$ theory, denoted by $D_q(\mathfrak{sl}_N, [N])$ in \cite{beem2025simplifying}, with $q=s$ and $N=r$. 

Indeed, from our dictionary, that theory corresponds to irregular curves with boundary data of standard type I AD$-A$ type with reduced parameter $\mathscr T=(1, \frac{q}{N}, [N])$. In that case we have $\mathcal C_0=\{\mathrm{Id}_N\}$, so the corresponding connections have trivial monodromy at 0, i.e. there is only an apparent singularity there. Furthermore, the irregular class at infinity is of the form $\Theta_\infty=\cir{q'}_\infty$, where $\Slope(q')=\frac{q}{N}=\frac{s}{r}$.

Some particular cases of this are related to some well-known mathematical objects:
\begin{itemize}
\item The case $r=2$, $s=3$ is related to the Airy equation $y''-xy=0$, whose irregular curve is $(\mathbb P^1, \infty, \cir{\frac{2}{3} z^{3/2}}_\infty)$. In that case the wild character variety is just a point, reflecting the fact that the Airy equation is \emph{rigid}.
 
\item For $s=5$ and $r=2$, the irregular curve $\bm\Sigma$ is the standard Lax representation for the Painlevé I equation. In that case the wild character variety has a complex dimension 2, and parametrizes solutions of Painlevé I. 
 \end{itemize}
 \item Consider the irregular curve \[
(\mathbb P^1, \infty, \cir{\lambda z^3}_\infty+\cir{\mu z^3}_\infty)
\]
with $\lambda, \mu\in \mathbb C^*$, $\lambda\neq \mu$. It corresponds to an irregular curve with boundary data of standard AD-$A$ type, with reduced parameter $\mathscr T=(2, 3, [2])$, i.e. to the Argyres--Douglas theory $D_3(\mathfrak{sl}_2, [2])$ in the notations of \cite{beem2025simplifying}. On the other hand, this irregular curve is none other but the standard Lax representation for the Painlevé II equation. 
\end{enumerate}
\end{example}

\section{Fourier transform and Möbius transformation}
\label{sec:Fourier}

In this section, we review some facts about the action of the Fourier transform on irregular curves with boundary data on $\mathbb P^1$, and discuss the action of the Möbius transformation exchanging 0 and $\infty$.

\subsection{Modified formal data} 

It turns out that the action of the Fourier transform on formal data of irregular connections on $\mathbb P^1$ is formulated more conveniently in terms of slightly modified formal data. The basic reason for this is that the trivial part of (the formalization of) the connection at the singularities at finite distance does not contribute to the Fourier transform. Concretely, this means that to obtain the formal data of the Fourier transform, one first needs to quotient by these trivial parts.

\begin{definition}
\label{def:truncation_conjuacy_class}
Let $\mathcal C\subset \mathrm{GL}_n(\mathbb C)$ be a conjugacy class, and $A\in \mathcal C$. The truncation $\tau(\mathcal C)$ of $\mathcal C$ is the conjugacy class of $A_{\vert{\mathrm{Im}(A-1)}}$ in $\mathrm{GL}_{m}(\mathbb C)$, where $m:=\mathrm{rk}(A-1)$.  (this does not depend on the choice of $A$.)
\end{definition}

This can be interpreted in terms of the corresponding Young diagrams. Let us introduce the following definition:

\begin{definition}~
\label{def:truncation_Young_diagram}
Let $[Y]$ be a Young diagram. The truncation $[\tau(Y)]$ is the Young diagram obtained from $[Y]$ by deleting its first column, if $[Y]$ is nonempty, or the empty Young diagram, if $[Y]$ is empty.
\end{definition}

It is then straightforward to see that taking the truncation of a conjugacy class amounts to taking the truncation of its Young diagram associated to the eigenvalue 1 (when nonempty):

\begin{lemma}
\label{lemma:truncation_in_terms_of_Young_diag}
Let $\mathcal C=\{(\lambda_1, [Y_{\lambda_1}]), \dots, (\lambda_k, [Y_{\lambda_k}])\}\subset \mathrm{GL}_n(\mathbb C)$ be a conjugacy class. 

\begin{itemize}
\item If the eigenvalues $\lambda_i$, for $i\in \{1, \dots, p\}$ are all different from $1$, then $\tau(\mathcal C)=\mathcal C$. 
\item Otherwise, if $\lambda_i=1$ for some $i\in \{1,\dots, p\}$, then we have
\[
\tau(\mathcal C)=\{(\lambda_1, [Y_{\lambda_1}]), \dots, (1, [\tau(Y_{\lambda_i}]), \dots, (\lambda_k, [Y_{\lambda_k}])\}.
\]
\end{itemize}
\end{lemma}

In particular, for a unipotent conjugacy class $\mathcal C=\{(1, [Y])\}\subset \mathrm{GL}_N(\mathbb C)$, we have $\tau(\mathcal C)=\{(1, [\tau(Y)])\}$.

Another immediate observation, which will be important, is the following converse:

\begin{lemma}
\label{lemma:reconstruction_conjugacy_class_from_truncation}
A conjugacy class $\mathcal C\subset\mathrm{GL}_n(\mathbb C)$ is fully determined by the datum of its truncation $\tau(\mathcal C)$ together with its rank $N$, namely if we write 
\[
\{(\lambda_1, [Y'_{\lambda_1}]), \dots, (\lambda_k, [Y'_{\lambda_k}])\}
\]
with $\lambda_i=1$ for some $i\in \{1, \dots, k\}$ (with $Y_1$ possibly empty), we have
\[
\mathcal C=\{(\lambda_1, [Y'_{\lambda_1}]), \dots, (1, [h,Y'_{1}]), \dots, (\lambda_k, [Y'_{\lambda_k}])\}
\]
where $h:=n-\sum_{i=1}^k \mathrm{rk}(Y'_{\lambda_i})$.
\end{lemma}

We introduce a variant of global formal data of irregular connections on $\mathbb P^1$, obtained by taking the truncations of the conjugacy classes of the formal monodromies of the regular parts  at finite distance.

\begin{definition}
Let $(\bm\Theta, \bm{\mathcal C})$ be (effective) global formal data, with
\[
\bm\Theta=\sum_{i=1}^p n_i\cir{q_i}, \qquad \bm{\mathcal C}=(\mathcal C_{\cir{q_1}}, \dots, \mathcal C_{\cir{q_p}}),
\]
with $\mathcal C_{\cir{q_i}}$ a conjugacy class in $\mathrm{GL}_{n_i}(\mathbb C)$ for $i\in \{1, \dots, p\}$.
The corresponding global \emph{modified formal data} are the pair $(\bm{\breve\Theta}, \breve{\bm{\mathcal C}})$ defined by:
\[
\bm{\breve\Theta}=\sum_{i=1}^p m_i\cir{q_i}, \qquad \breve{\bm{\mathcal C}}=(\breve{\mathcal C}_{\cir{q_1}}, \dots, \breve{\mathcal C}_{\cir{q_p}}),
\]
where
\begin{itemize}
\item If $\cir{q_i}$ is a tame circle at finite distance, i.e. $\cir{q_i}=\cir{0}_a$ with $a\neq\infty$, then $m_i=\mathrm{rk}(A-1)$ for any $A\in \mathcal C_{\cir{q_i}}$, and $\breve{\mathcal C_i}:=\tau(\mathcal C_i)\subset \mathrm{GL}_{m_i}(\mathbb C)$.
\item Otherwise, $m_i:=n_i$ and $\breve{\mathcal C_i}=\mathcal C_i$.
\end{itemize}
\end{definition}

Notice that, if $a_1, \dots, a_n$ are the singularities of $\bm{\Theta}$ at finite distance, and if we view the modified global irregular class $\bm{\breve{\Theta}}$ as a collection $(\breve\Theta_{a_1}, \dots, \breve\Theta_{a_n}, \breve\Theta_{\infty})$ of local irregular classes (we can always include $\infty$ as a singularity, up to taking $\breve\Theta_\infty$ as the trivial irregular class $N \cir{0}_\infty$), then now in general we do not have any more the equality of ranks $\mathrm{rk} (\Theta_{a_i})=\mathrm{rk}(\Theta_\infty)$ for $i\in \{1, \dots, n\}$, but only the inequality
\[
\mathrm{rk}(\Theta_{a_i})\leq\mathrm{rk}(\Theta_\infty).
\]

Furthermore, since $\mathrm{rk}(\Theta_\infty)=\mathrm{rk}(\bm\Theta)$, it follows from Lemma \ref{lemma:reconstruction_conjugacy_class_from_truncation} that the non-modified irregular class $\bm{\Theta}$ is fully determined by the modified one $\bm{\breve{\Theta}}$, i.e. modified and non-modified global formal data are equivalent (cf. \cite[Corollary 3.6]{arinkin2008fourier}.)

\begin{remark}
\label{remark:link_connection_D_modules}
From a more conceptual point of view, the passage to modified formal data has to do with the notion of nearby and vanishing cycles for regular holonomic $\mathcal D$-modules and perverse sheaves in dimension 1, cf. \cite{arinkin2008fourier}. More precisely, given a connection $(E,\nabla)$ on $\mathbb A^1\smallsetminus\{a_1, \dots a_m\}$, there one can associate to it a $\mathcal D_{\mathbb A^1}$-module $M$ called its minimal extension (or middle extension, or intermediate extension). Taking the non-modifed formal data amounts to considering the formal nearby cycles of $M$ at $a_i$, while taking the modified formal data amounts to considering their formal vanishing cycles. 
\end{remark}

\subsection{Fourier transform}

The Weyl algebra $A_1=\mathbb C[z, \partial_z]$ is the algebra of differential operators on the affine line $\mathbb A^1(\mathbb C)$. The Fourier (or Laplace, or Fourier-Laplace) transform is induced by the automorphism of $A_1$ defined by
\begin{equation}
\label{eq:def_Fourier_transform}
\left\lbrace
\begin{array}{ccc}
z &\mapsto &-\partial_z,\\
\partial_z & \mapsto & z.
\end{array}
\right.
\end{equation}
 
In turn, if $M$ is an $A_1$-module, its Fourier transform is defined as the $A_1$-module $M'$ with the same underlying $\mathbb C$-vector space, and $A_1$-action induced by \eqref{eq:def_Fourier_transform}.

Due to the fact that algebraic connections on Zariski open subsets of $\mathbb P^1(\mathbb C)$ are closely related to $A_1$-modules (cf. Remark \ref{remark:link_connection_D_modules}), up to some caveats this also induces a well-defined operation on connections: if $(E,\nabla)$ is an irreducible connection on a Zariski open set $U$ that is not a rank one connection with just a pole of order 2 at infinity, there is a well-defined irreducible connection $(E', \nabla')$ on a Zariski open subset $U'$ of $\mathbb P^1$ that we call the Fourier transform of $(E,\nabla)$, and denote by $F\cdot (E,\nabla)$. 

The Fourier transform acts in a very nontrivial way, typically changing the rank, number of singularities, and pole orders of the connections. On the Betti side, it is known in some cases how to explicitly determine the Stokes data of the Fourier transform in terms of those of the initial connection, but in general it is a difficult problem.

However, at the level of formal data, the action of the Fourier transform is well-understood: the stationary phase formula \cite{malgrange1991equations,garcia2004microlocalization, fang2009calculation, sabbah2008explicit,graham2013calculation} implies that the formal data of the Fourier transform $F\cdot (E, \nabla)$ are fully determined by those of $(E, \nabla)$:

\begin{theorem}
There exists a well-defined self-bijection of the set of effective formal data of connections on $\mathbb P^1$, that we call the formal Fourier transform, and will also denote by $F$, such that if $(E,\nabla)$ has formal data $(\bm\Theta, \bm{\mathcal C})$, then $F\cdot (E,\nabla)$ has formal data $F\cdot (\bm\Theta, \bm{\mathcal C})$.
 
Furthermore, the formal Fourier transform has the following form: there is a self-bijection of the set $\pi_0(\mathcal I)$ of all Stokes circles (over all points in $\mathbb P^1)$, that we will denote also by $F$, such that if $(\bm{\breve\Theta}, \bm{\breve{\mathcal C}})$ are the modified formal data associated to $(\bm\Theta, \bm{\mathcal C})$, the modified formal data associated to $F\cdot(\bm\Theta, \bm{\mathcal C})$ are given by:

\begin{equation}
F\cdot (\bm{\breve\Theta}, \bm{\breve{\mathcal C}}):=\left\{ \left(F\cdot \cir{q_i}, (-1)^{\Irr(q_i)}\breve{\mathcal C}_{\cir{q_i}}\right) \;\middle \vert\; i\in \{1, \dots, p\} \right\} .
\end{equation} 
Moreover, the bijection $F$ on Stokes circles is given explicitly by a Legendre transform.
\end{theorem}

Concretely, this implies that to determine the formal data of the Fourier transform, one proceeds in three steps, as follows:
\begin{enumerate}
\item Pass to the modified formal data
\item Take the Legendre transform to obtain the modified formal data of the Fourier transform
\item Pass to the corresponding the non-modified formal data
\end{enumerate}

For our purposes here, it will not be necessary to discuss in detail the Legendre transform for all types of Stokes circles.

\begin{lemma}For Stokes circles at zero and infinity, we have:
\begin{itemize}
\item $F\cdot \cir{0}_0=\cir{0}_\infty$.
\item If $\cir{q}$ is a Stokes circle at 0, of slope $k=\frac{s}{r}\neq 0$, then $F\cdot \cir{q}$ is a Stokes circle at $\infty$ of slope $\frac{k}{k+1}=\frac{s}{r+s}$.
\item If $\cir{q}$ is a Stokes circle at $\infty$ of slope $k=\frac{s}{r}>1$, then $F\cdot \cir{q}$ is a Stokes circle at $\infty$ of slope $\frac{k}{k-1}=\frac{s}{s-r}>1$.
\item If $\cir{q}$ is a Stokes circle at $\infty$ of slope $k=\frac{s}{r}<1$, then $F\cdot \cir{q}$ is a Stokes circle at $0$ of slope $\frac{k}{1-k}=\frac{s}{r-s}$.
\end{itemize}
\end{lemma}

\begin{proof}
This follows from the stationary phase formula, using the explicit form of the Legendre transform.
\end{proof}

\subsection{Möbius transformation}

We now briefly discuss the Möbius transformation $M$ exchanging $0$ and $\infty$. Let $\varphi:\mathbb P^1\to \mathbb P^1$ the automorphism of $\mathbb P^1$ defined by $z\mapsto\frac{1}{z}$.  

\begin{definition}
Let $(E,\nabla)$ be a connection on a Zariski open subset $U$ of $\mathbb P^1$. We define $M\cdot (E,\nabla)$ as the connection $\varphi_*(E,\nabla)$ on $\varphi(U)$.
\end{definition}

Concretely, applying $M$ just amounts to do the change of coordinate on $\mathbb P^1$ given by $z'=\frac{1}{z}$. In turn, it is straightforward to see that, if $\cir{q}$ is a Stokes circle at $a\in \mathbb P^1$, there is a well-defined Stokes circle $M\cdot \cir{q}$ at $\varphi(a)$, obtained as the connected component of the polar part of $\varphi_*(q)$. Moreover $M\cdot \cir{q}$ has the same ramification order, irregularity and slope as $\cir{q}$.

The formal data of $M\cdot (E,\nabla)$ are related to those of $(E,\nabla)$ as follows:

\begin{lemma}
Let $(E,\nabla)$ be a connection on a Zariski open subset of $\mathbb P^1$, with irregular curve with boundary data $\bm\Sigma=(\mathbb P^1, \mathbf a, \bm\Theta, \bm{\mathcal C})$. 
Then the irregular curve with boundary data of $M\cdot (E,\nabla)$ only depends on $\bm\Sigma$, and is given by
\[
M\cdot \bm\Sigma:=(\mathbb P^1, \varphi(\mathbf a), M\cdot(\bm\Theta, \bm{\mathcal C}))
\]
where, for $(\bm\Theta, \bm{\mathcal C})=\left\lbrace \left(\cir{q_i}, \mathcal C_{\cir{q_i}}\right)\;\middle\vert\; i\in \{1, \dots, p\}\right\rbrace$, we set

\[
M\cdot(\bm\Theta, \bm{\mathcal C}):=\left\lbrace\left(M\cdot\cir{q_i}, \mathcal C_{\cir{q_i}}\right)\;\middle\vert\; i\in \{1, \dots, p\}\right\rbrace.
\]
\end{lemma}

\section{Argyres--Douglas dualities from operations on connections in type I}
\label{sec:type_I_dualities}

\subsection{Elementary AD dualities in type I}

We are ready to study the transformation of irregular curves with boundary data of generalized type I AD-$A$ type under compositions of $F$ and $M$. 

In general, not any composition of $M$ and $F$ will preserve the property of being an irregular curve with boundary data of generalized type I AD-$A$ type. For this reason, we introduce some elementary compositions for which this will be the case.

\begin{definition}
The type I elementary AD-$A$-operations are the following operations on irregular curves with boundary data on $\mathbb P^1$:
\begin{itemize}
\item the Fourier transform $F$;
\item the composition $\widetilde F^+:=FM$;
\item the composition $\widetilde F^{-}:=MF$.
\end{itemize}
\end{definition}

\begin{definition}
Let $O\in \{F, \widetilde F^+, \widetilde F^+\}$ be an elementary type I AD-$A$-operation, and $\mathcal T=(m, k, [Y^0], [Y^{\infty}])$ a generalized type I AD-$A$ parameter. We say that $O$ is an \emph{allowed} operation on $\mathcal T$ in the following cases:
\begin{enumerate}
\item $O=F$ and $k>1$;
\item $O=\widetilde F^+$;
\item $O=\widetilde F^-$ and $k<1$.
\end{enumerate} 
\end{definition}

Our first main observation is that irregular curves with boundary data that are of generalized type I AD-$A$ type remain so under allowed elementary AD-$A$-operations, justifying our terminology.

\begin{proposition}
\label{prop:transformation_parameter_elementary_operation_type_I}
Let $\bm\Sigma$ be an irregular curve with boundary data of generalized type I AD-$A$ type, with parameter $\mathcal T=(m,k,[Y^0], [Y^\infty])$. Then, if $O$ is an allowed elementary AD-$A$-operation on $\mathcal T$, the irregular curve with boundary data $O\cdot \bm\Sigma$ is of generalized AD-$A$ type. 

Furthermore, in each case, we have the following explicit expression for the parameter $O\cdot \mathcal T$ of $O\cdot\bm\Sigma$:
\begin{enumerate}
\item If $k>1$, then the irregular curve with boundary data $F\cdot \bm\Sigma$ is  of generalized type I AD-$A$ type, with parameter
\begin{equation}
F\cdot \mathcal T=\left(m,\frac{s}{s-r},[ms-h_1(Y^0),Y^{\infty}], [\tau(Y^{0})]\right).
\end{equation}
\item The irregular curve with boundary data $\widetilde F^+\cdot \bm\Sigma$ is of generalized type I AD-$A$ type, with parameter
\begin{equation}
\widetilde F^+\cdot\mathcal T=\left(m,\frac{s}{s+r},[ms-h_1(Y^\infty),Y^{0}], \tau(Y^{\infty})\right).
\end{equation}
\item If $k<1$, the irregular curve with boundary data $\widetilde F^-\cdot \bm\Sigma$ is of  generalized type I AD-$A$ type, with parameter
\begin{equation}
\widetilde F^-\cdot\mathcal T=\left(m,\frac{s}{r-s},[\tau(Y^{0})], [ms-h_1(Y^0),Y^{\infty}]\right). 
\end{equation}
\end{enumerate}
\end{proposition}

Notice that in every case, the pattern is as follows: applying one elementary step removes the first column of one of the Young diagrams $[Y^0]$, $[Y^\infty]$, and adds the `complement' of that column to the other Young diagram. Which one of $[Y^0]$ and $[Y^\infty]$ decreases depends on the elementary operation.
\begin{proof}
This follows quite directly from the stationary phase formula, applied to irregular curves with boundary data of generalized $AD$-A type. Let us give the details.

By definition, an irregular curve with boundary data $\bm\Sigma=(\bm\Theta, \bm{\mathcal C})$ of generalized type I AD-$A$ type with parameter $\mathcal T$ corresponds to a global irregular class of the form
\[
\bm{\Theta}=\cir{q_1}_\infty+\dots+\cir{q_m}_\infty+ n_{\cir{0}_\infty} \cir{0}_\infty+ N\cir{0}_0,
\]
where $N=\mathrm{rk}(\bm\Theta)=\mathrm{rk}(Y^0)$, $n_{\cir{0}_\infty}=\mathrm{rk}(Y^\infty)$, and the boundary data include the conjugacy classes $\mathcal C_{\cir{0}_\infty}=\{(1, [Y^\infty])\}\subset \mathrm{GL}_{n_{\cir{0}_\infty}}(\mathbb C)$ and $\mathcal C_{\cir{0}_0}=\{(1, [Y^0])\}\subset \mathrm{GL}_N(\mathbb C)$. In particular we have the relation 
\[
N=mr+n_{\cir{0}_\infty}.
\]

Let us discuss the first case. Assume $k>1$. The modified formal data $(\bm{\breve\Theta},\bm{\breve\mathcal C})$ associated to $\bm\Sigma=(\bm\Theta,\bm{\mathcal C})$ feature the modified global irregular class
\[
\bm{\breve\Theta}=\cir{q_1}_\infty+\dots+\cir{q_m}_\infty+ n_{\cir{0}_\infty} \cir{0}_\infty+ \breve N_0\cir{0}_0,
\]
where $\breve N_0=\mathrm{rk}(\tau(\mathcal C_0))=N-h_1(Y^0)$, and include the modified conjugacy classes $\breve{\mathcal C}_{\cir{0}_0}=\tau(\mathcal C_0)=\{(1, [\tau(Y^0_\infty)])\}$, and $\breve{\mathcal C}_{\cir{0}_\infty}=\mathcal C_\infty$. 
The modified formal data $(\bm{\breve\Theta}', \bm{\breve\mathcal C}')$ associated to $F\cdot \Sigma$ are thus given by the modified irregular class
\[
\bm{\breve\Theta}'=\cir{\widetilde q_1}_\infty+\dots+ \cir{\widetilde q_m}_\infty+ n_{\cir{0}_\infty} \cir{0}_0+ \breve N_0\cir{0}_\infty
\]
where $\widetilde{q_i}$ is of slope $\frac{s}{s-r}$ for $i\in \{1, \dots, m\}$, and $\bm{\breve\mathcal C}'$ includes the conjugacy classes  
\begin{align*}
\breve{\mathcal C}'_{\cir{0}_0} &= {\mathcal C}_\infty \\
\breve{\mathcal C}'_{\cir{0}_\infty} &= \breve{\mathcal C}_0=\tau(\mathcal C_0).
\end{align*}
In particular, the rank of the Fourier transform $F\cdot \bm{\Sigma}$ is given by
\[
N'=\mathrm{rk}(F\cdot \bm{\Sigma})=m(s-r)+\breve N_0=ms+\mathrm{rk}(Y^\infty)-h_1(Y^0)
\]
Finally, the global formal data $(\bm{\Theta}',\bm{\mathcal C}')$ of $F\cdot \mathbf \Sigma$ are the corresponding non-modified formal data:  they are obtained from $(\bm{\breve\Theta}', \bm{\breve\mathcal C}')$ by replacing the unipotent conjugacy class $\breve{\mathcal C}'_{\cir{0}_0}=\{(1, [Y^\infty])\}$ by $\mathcal C'_{\cir{0}_0}:=\{(1, [h,Y^\infty])\}$, where $h=N'-\mathrm{rk}(Y^\infty)=ms-h_1(Y^0)$, so we obtain the desired result.

Let us now discuss the second case. First, applying the Möbius transformation $M$ has the effect on $\mathbf \Sigma$ of exchanging the roles of $0$ and $\infty$, hence the formal data $(\bm{\underline{\Theta}}, \bm{\underline{\mathcal C}})$ of $M\cdot \bm\Sigma$ are of the form
\[
\bm{\underline\Theta}=\cir{\underline q_1}_0+\dots+ \cir{\underline q_m}_0+ n_{\cir{0}_\infty} \cir{0}_0+ N\cir{0}_\infty,
\]
and the conjugacy classes $\underline{\mathcal C}_{\cir{0}_0}=\mathcal C_\infty$, $\underline{\mathcal C}_{\cir{0}_\infty}=\mathcal C_0$. 
Next, the modified formal data $(\underline{\bm{\breve{\Theta}}},\underline{\bm{\breve{\mathcal C}}})$ associated to $(\bm{\underline{\Theta}}, \bm{\underline{\mathcal C}})$ satisfy
\[
\bm{\breve{\underline\Theta}}=\cir{\underline q_1}_0+\dots+ \cir{\underline q_m}_0+ n' \cir{0}_0+ N\cir{0}_\infty,
\]
with $\cir{\underline q_i}$ of slope $\frac{s}{r}$ for $i\in \{1, \dots, m\}$,  $n'=\mathrm{rk}(\tau(\mathcal C_\infty))$, $\underline{\breve{\mathcal C}}_{\cir{0}_0}=\tau(\mathcal C_\infty)$, and $\underline{\breve{\mathcal C}}_{\cir{0}_\infty}=\mathcal C_0$. 
It follows that the formal data $(\bm{\breve\Theta}', \bm{\breve{\mathcal C}}')$ of $FM\cdot \bm\Sigma$ are of the form
\[
\bm{\breve{\Theta}}'=\cir{q'_1}_\infty+\dots+\cir{q'_m}_\infty+ n' \cir{0}_\infty+ N\cir{0}_0,
\]
with $q'_i$ of slope $\frac{s}{s+r}$ for $i\in\{1, \dots, m\}$, and $\breve{\mathcal C}'_{\cir{0}_\infty}=\tau(\mathcal C_\infty)$, and $\breve{\mathcal C}'_{\cir{0}_0}=\mathcal C_0$. 
In particular, the rank $N'$ of $FM\cdot \bm\Sigma$ is given by
\[
N'=m(r+s)+\mathrm{rk}(\breve{\mathcal C}'_{\cir{0}_\infty})=m(r+s)+\mathrm{rk}(\tau(\mathcal C_\infty))=ms +\mathrm{rk}(Y^0)-h_1(Y^\infty).
\]
Finally, the formal data $(\bm\Theta', \bm{\mathcal C}')=FM\cdot \bm\Sigma$ are the corresponding non-modified formal data: there are obtained by replacing $\breve{\mathcal C}'_{\cir{0}_0}=\mathcal C_0$ by $\mathcal C'_{\cir{0}_0}:=\{(1,[h,Y^0])\}$, with $h=N'-\mathrm{rk}(Y^0)=ms-h_1(Y^\infty)$, which yields the desired result.

Let us discuss the third case. Assume $k<1$. The modified formal data $(\bm{\breve\Theta},\bm{\breve\mathcal C})$ associated to $\bm\Sigma=(\bm\Theta,\bm{\mathcal C})$ feature the modified global irregular class
\[
\bm{\breve\Theta}=\cir{q_1}_\infty+\dots+\cir{q_m}_\infty+ n_{\cir{0}_\infty} \cir{0}_\infty+ \breve N_0\cir{0}_0,
\]
where $\breve N_0=\mathrm{rk}(\tau(\mathcal C_0))$, and the modified conjugacy classes $\breve{\mathcal C}_{\cir{0}_0}=\tau(\mathcal C_0)=\{(1, [\tau(Y^0_\infty)])\}$, and $\breve{\mathcal C}_{\cir{0}_\infty}=\mathcal C_\infty$. 
The modified formal data $(\bm{\breve{\underline\Theta}}, \bm{\breve{\underline{\mathcal C}}})$ associated to $F\cdot \bm\Sigma$ are given by the modified irregular class
\[
\bm{\breve{\underline\Theta}}=\cir{\underline q_1}_0+\dots+ \cir{\underline q_m}_0+ n_{\cir{0}_\infty} \cir{0}_0+ \breve N_0\cir{0}_\infty
\]
where $\underline q_i$ is of slope $\frac{s}{r-s}$ for $i\in \{1, \dots, m\}$, and $\bm{\breve{\underline{\mathcal C}}}$ includes the conjugacy classes  
\begin{align*}
\breve{\underline{\mathcal C}}_{\cir{0}_0} &= {\mathcal C}_\infty \\
\breve{\underline{\mathcal C}}_{\cir{0}_\infty} &= \breve{\mathcal C}_{\cir{0}_0}=\tau(\mathcal C_0).
\end{align*}
In particular, the rank of the Fourier transform $F\cdot \bm{\Sigma}$ is given by
\[
N'=\mathrm{rk}(F\cdot \bm{\Sigma})=\breve N_0=\mathrm{rk}(Y^0)-h_1(Y^0).
\]
Next, the global formal data $(\bm{\underline\Theta},\bm{\underline{\mathcal C}})$ of $F\cdot \mathbf \Sigma$ are the corresponding non-modified formal data:  they are obtained from $(\bm{\breve{\underline\Theta}}, \bm{\breve{\underline{\mathcal C}}})$ by replacing the unipotent conjugacy class $\breve{\underline{\mathcal C}}_{\cir{0}_0}$ by $\underline{\mathcal C}_{\cir{0}_0}:=\{(1,[h,Y^\infty_1])\}$, 
where 
\[
h=N'-\mathrm{rk}(\underline{\breve\Theta}_0)=N'-(m(r-s)+\mathrm{rk}(Y^\infty))=ms-h_1(Y^0).
\]
Finally, the (non-modified) irregular class of $MF\cdot \bm\Sigma$ is obtained from $(\bm{\underline\Theta},\bm{\underline{\mathcal C}})$ by exchanging the roles $0$ and $\infty$, which gives the desired result.
\end{proof}

\subsection{Orbits under elementary transformations}

Let us now study what happens when applying a composition of allowed elementary operations. 

\begin{definition}

Let $\bm\Sigma$ be an irregular curve with boundary data of generalized type I AD-$A$ type. We define $\mathcal O(\bm\Sigma)$ as the set of all irregular curves with boundary data of generalized type I AD-$A$ type of the form $O_k\dots O_1\cdot \bm\Sigma$, where $k\geq 1$, and for each $i\in \{1, \dots, k\}$, $O_i$ is an allowed elementary  type I AD-operation on $O_{i-1}\cdot O_1\cdot \bm\Sigma$. 

Similarly, if $\mathcal T$ is a generalized type I AD-$A$ parameter we define in the same way a set $\mathcal O(\mathcal T)$.
\end{definition}

By a slight abuse of language, we will call $\mathcal O(\bm\Sigma)$ (resp. $\mathcal O(\mathcal T)$) the orbit of $\bm\Sigma$ (resp. $\mathcal T$) under AD transformations.

\begin{theorem}
Let $\mathcal T=(m, k, \mathcal C_0, \mathcal C_\infty)$ be a generalized type I AD-$A$ parameter. Let us write $k=\frac{s}{r}$, with $s, r$ coprime, and assume that $s>1$. Let $r=\kappa s+\rho$ , with $\kappa\in \mathbb Z_{\geq 0}$ and $\rho\in \{1, \dots, s-1\}$ be the euclidean division of $r$ by $s$. 
Then the orbit $\mathcal O(\bm\Sigma)$ has the structure described on Fig. \ref{fig:structure_orbit} below.

\begin{figure}[h]
\centering
\begin{tikzpicture}
\tikzstyle{vertex}=[circle,fill=black,minimum size=6pt,inner sep=0pt]
\node[vertex] (A0) at (0,0){} ;
\node[vertex] (B0) at (3,0){};  

\draw (-1.5,0) node {$\frac{s}{\rho}$};
\draw (4.5,0) node {$\frac{s}{s-\rho}$};

\draw (-0.3,-0.3) node {$\mathcal T_+$};
\draw (3.3,-0.3) node {$\mathcal T_-$};

\node[vertex] (A1) at (0,1.5)  {};
\node (A2) at (0,2.4)  {$\vdots$};
\node[vertex] (A3) at (0,3) {};
\node[vertex] (A4) at (0,4.5)  {};
\node[vertex] (A5) at (0,6)  {};
\node (A6) at (0,6.7)  {$\vdots$};

\draw (-1.5,1.5) node  {$\frac{s}{s+\rho}$};
\draw  (-1.5,2.4)  node {$\vdots$};
\draw (-1.5,3)  node {$\frac{s}{(\kappa-1) s+\rho}$};
\draw (-1.5,4.5)  node {$\frac{s}{\kappa s+\rho}$};
\draw (-0.5,4.5)  node {$\mathcal T$};

\node[vertex] (B1) at (3,1.5)  {};
\node (B2) at (3,2.4)  {$\vdots$};
\node[vertex] (B3) at (3,3) {};
\node[vertex] (B4) at (3,4.5)  {};
\node[vertex] (B5) at (3,6)  {};
\node (B6) at (3,6.7)  {$\vdots$};

\draw (4.5,1.5) node  {$\frac{s}{s-\rho}$};
\draw  (4.5,2.4)  node {$\vdots$};
\draw (4.5,3)  node {$\frac{s}{(\kappa-1) s-\rho}$};
\draw (4.5,4.5)  node {$\frac{s}{\kappa s-\rho}$};

\draw[red, <-] (A0) to[bend left] node[midway, below]{$F$} (B0);
\draw[blue,->] (A0) to[bend right] node[midway, below]{$F$} (B0);

\draw[red,->] (A0) to[bend right] node[midway, right] {$\widetilde F^+$} (A1);
\draw[blue,->] (A1) to[bend right] node[midway, left] {$\widetilde F^-$} (A0);

\draw[red, ->] (A3) to[bend right] node[midway, right] {$\widetilde F^+$} (A4);
\draw[blue,->] (A4) to[bend right] node[midway, left] {$\widetilde F^-$} (A3);

\draw[red, ->] (A4) to[bend right] node[midway, right] {$\widetilde F^+$} (A5);
\draw[blue,->] (A5) to[bend right] node[midway, left] {$\widetilde F^-$} (A4);

\draw[blue, ->] (B0) to[bend right] node[midway, right] {$\widetilde F^+$} (B1);
\draw[red,->] (B1) to[bend right] node[midway, left] {$\widetilde F^-$} (B0);

\draw[blue,->] (B3) to[bend right] node[midway, right] {$\widetilde F^+$} (B4);
\draw[red, ->] (B4) to[bend right] node[midway, left] {$\widetilde F^-$} (B3);

\draw[blue,->] (B4) to[bend right] node[midway, right] {$\widetilde F^+$} (B5);
\draw[red, ->] (B5) to[bend right] node[midway, left] {$\widetilde F^-$} (B4);

\draw (-1.5,-1) node {$r'\equiv \rho \mod s$};
\draw (4.5,-1) node {$r'\equiv -\rho \mod s$};
\draw[dashed] (-3, 0.38) -- (7,0.38);
\draw (6, 0) node {$k'>1$};
\draw (6, 0.7) node {$k'<1$};
\end{tikzpicture}
\caption{Structure of the orbit $\mathcal O(\mathcal T)$, for $s>1$. The vertices correspond to the elements of the orbit, and the arrows correspond to elementary type I AD-$A$ operations. The blue arrows correspond to the steps where the Young diagram coming from $[Y^0]$ decreases and the one coming from $[Y^\infty]$ increases, and the red ones to the opposite steps.  On the bottom row, the Fourier transform exchanges the Young diagrams at 0 and $\infty$.  The slopes are all of the form $k'=\frac{s}{r'}$, with $r'\equiv \pm \rho \mod s$. The oriented arrows correspond to applying an elementary AD-operation. The left column corresponds to the case $r'\equiv \rho \mod s$, while the right one corresponds to the case $r'\equiv -\rho \mod s$. The bottom row corresponds to the only two slopes satisfying $k'>1$.}
\label{fig:structure_orbit}
\end{figure}
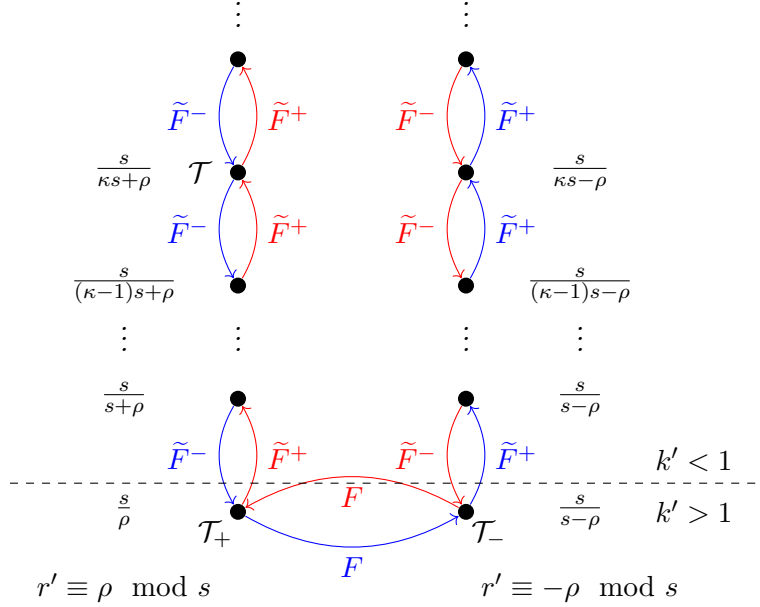

In particular, the set of all slopes of elements in the orbit $\mathcal O(\mathcal T)$ only depends on $k$, and is given by
\begin{equation}
\label{eq:form_slopes_orbit}
\mathcal O(k):=\left\lbrace \frac{s}{ls\pm \rho} \;\middle \vert\; l\in \mathbb Z_{\geq 0}, ls\pm \rho>0 \right\rbrace.
\end{equation}

Furthermore, we have the following: 
\begin{itemize}
\item If $s>2$, for any $k'\in \mathcal O(k)$, there exists a unique element $\mathcal T'\in \mathcal O(T)$ with slope $k'$.
\item If $s=2$, then:
\begin{itemize}
\item either (the generic case), for each $k'\in \mathcal O(k)$,  there exist exactly two elements $\mathcal T'\in \mathcal O(T)$ with slope $k'$;
\item or (the `symmetric' case) for each $k'\in \mathcal O(k)$
or there exists exactly one element $\mathcal T'\in \mathcal O(T)$ with slope $k'$.
\end{itemize}
\end{itemize}
\end{theorem}

\begin{proof}
That the situation is as represented on Fig. \ref{fig:structure_orbit} follows directly by induction from the formulas for the transformation of the slope of the wild Stokes circles in Prop. \ref{prop:transformation_parameter_elementary_operation_type_I} and Prop. \ref{prop:transf_parameter_elementary_transf_twist}, namely: 
\begin{itemize}
\item $F$ is allowed when $\frac{s}{r}>1$ and induces on the slope the transformation $\frac{s}{r}\mapsto \frac{s}{s-r}>1$.
\item $\widetilde{F}^{+}$ is always allowed and induces $\frac{s}{r}\mapsto \frac{s}{s+r}<1$.
\item $\widetilde{F}^{-}$ is allowed when $\frac{s}{r}<1$ and induces then $\frac{s}{r}\mapsto \frac{s}{r-s}$.
\end{itemize}
 
Now, if $s>2$, we have $\rho\neq s-\rho$. In turn, there is a unique element $\mathcal T'_+\in \mathcal O(\mathcal T)$ of slope $\frac{s}{\rho}$, given by  $\mathcal T'_+:=(\widetilde F^{-})^\kappa \cdot \mathcal T$ and a unique element $\mathcal T'_-\in \mathcal O(\mathcal T)$ of slope $\frac{s}{\rho}$, given  by $\mathcal T'_-:=F\cdot \mathcal T'_+$.

Moreover, for any $k\geq 0$, there is a unique element in $\mathcal O(\mathcal T)$ of slope $\frac{s}{ks+\rho}$, given by $(\widetilde F^+)^k\cdot \mathcal T'_+$, and a unique element in $\mathcal O(\mathcal T)$ of slope $\frac{s}{ks-\rho}$, given by $(\widetilde F^+)^k\cdot \mathcal T'_-$.

If $s=2$, then necessarily $\rho=1=s-\rho$, so there are two elements $\mathcal T'_+$ and $\mathcal T'_-$ of slope $2$ exchanged by the Fourier transform, given by $\mathcal T'_+:=(\widetilde F^{-})^\kappa \cdot \mathcal T$ and $\mathcal T'_-:=F\cdot \mathcal T'_1$. 

Now, if $\mathcal T'_+$ and $\mathcal T'_-$ are different, for any $k\geq 0$, there are two elements in $\mathcal O(\mathcal T)$ of slope $\frac{2}{2k+1}$, given by $\widetilde F^+(\mathcal T'_+)$ and $\widetilde F^+(\mathcal T'_-)$. Otherwise, if $\mathcal T'_+=\mathcal T'_-$, then for any $k\geq 0$, there is a unique element in $\mathcal O(\mathcal T)$ of slope $\frac{2}{2k+1}$, given by $\widetilde F^+(\mathcal T'_+)$.
\end{proof}

The case $s=1$ is somewhat different:

\begin{lemma}
Let $\mathcal T=(m, k, \mathcal C_0, \mathcal C_\infty)$ be a generalized type I AD-$A$ parameter, such that $k=\frac{1}{r}$, with $r\geq 1$. 
Then for any $l\geq 1$, there is a unique element $\mathcal T'\in \mathcal O(\mathcal T)$ with slope $\frac{1}{l}$.
\end{lemma}

The situation is represented on Fig. \ref{fig:structure_orbit_s=1}.

\begin{figure}[h]
\centering
\begin{tikzpicture}
\tikzstyle{vertex}=[circle,fill=black,minimum size=6pt,inner sep=0pt]

\draw (-1.5,0) node {$1$};

\draw (-0.3,-0.3) node {$\mathcal T_1$};

\node[vertex] (A0) at (0,0)  {};
\node[vertex] (A1) at (0,1.5)  {};
\node (A2) at (0,2.4)  {$\vdots$};
\node[vertex] (A3) at (0,3) {};
\node[vertex] (A4) at (0,4.5)  {};
\node[vertex] (A5) at (0,6)  {};
\node (A6) at (0,6.7)  {$\vdots$};

\draw (-1.5,1.5) node  {$\frac{1}{2}$};
\draw  (-1.5,2.4)  node {$\vdots$};
\draw (-1.5,3)  node {$\frac{1}{r-1}$};
\draw (-1.5,4.5)  node {$\frac{1}{r}$};
\draw (-0.5,4.5)  node {$\mathcal T$};

\draw[red,->] (A0) to[bend right] node[midway, right] {$\widetilde F^+$} (A1);
\draw[blue,->] (A1) to[bend right] node[midway, left] {$\widetilde F^-$} (A0);

\draw[red, ->] (A3) to[bend right] node[midway, right] {$\widetilde F^+$} (A4);
\draw[blue,->] (A4) to[bend right] node[midway, left] {$\widetilde F^-$} (A3);

\draw[red, ->] (A4) to[bend right] node[midway, right] {$\widetilde F^+$} (A5);
\draw[blue,->] (A5) to[bend right] node[midway, left] {$\widetilde F^-$} (A4);

\end{tikzpicture}
\caption{Structure of the orbit $\mathcal O(\mathcal T)$, for $s=1$. At $\mathcal T_1$, applying $F$ yields an irregular curve with boundary data with $m+2$ regular singularities.}
\label{fig:structure_orbit_s=1}
\end{figure}
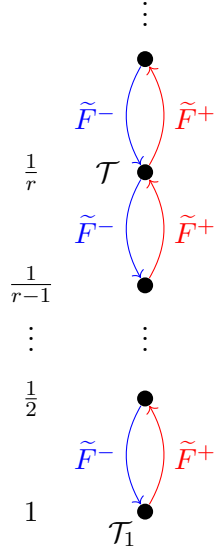

\begin{proof}
Compared to Fig. \ref{fig:structure_orbit}, in the case $s=1$ one still has the $k'>1$ part of the left side of the figure: for any $l\geq 1$, we obtain an element $\mathcal T'\in \mathcal O(\mathcal T)$ with slope $\frac{1}{l}$, by taking $\mathcal T'=(\widetilde F^+)^{l-r}$ if $l\geq r$, and $\mathcal T':=(\widetilde F^-)^{r-l}$, for $1\leq l\leq r$.

The difference  with the $s>1$ case is that now, at the parameter $\mathcal T_1=(\widetilde F^-)^{r-1}$, we are at slope $1$, so it is not allowed to take the blue arrow any more. Indeed, for an irregular curve $\bm\Sigma$ with boundary data with parameter $\mathcal T_1$, the wild Stokes circles are of slope 1, i.e. of the form $\cir{q_i}_\infty=\cir{a_i z}_\infty$, with $a_i\in \mathbb C^*$ for $i\in \{1, \dots, m\}$. In turn, the Fourier transform $F\cdot \bm\Sigma'$ is no longer of generalized AD-$A$ type. Indeed by the stationary phase formula we have $F\cdot \cir{a_i z}_\infty=\cir{0}_{a_i}$, so $F\cdot \bm\Sigma'$ has $m+2$ regular singularities, at the points $0,\infty, a_1, \dots, a_m$.  
\end{proof}

\subsection{Dualities between standard parameters}

We can now state our first main results: the dualities of \cite{beem2025simplifying} in the type I case can be all expressed as compositions of elementary type I AD-$A$ operations $F, \widetilde F^+$, and $\widetilde F^+$, i.e. they correspond to some particular paths in Fig. \ref{fig:structure_orbit}.

\begin{proposition}
\label{prop:duality_I}
Let $\bm\Sigma$ be an irregular curve with boundary data of standard type I AD-$A$ type, with reduced parameter $\mathscr T=(m,k,[Y])$. 
Then for any integer $l\geq 0$, $(\widetilde{F}^+)^l\cdot\Sigma
$ is also of standard type I AD-$A$ type, with reduced parameter
\begin{equation}
\label{eq:main_duality_add_column}
(\widetilde{F}^+)^l\cdot\mathscr T:=\left(m, \frac{s}{r+ls}, [(ms)^l,Y]\right)
\end{equation}
\end{proposition}

\begin{proof}
This follows directly by induction from applying the second part of Prop. \ref{prop:transformation_parameter_elementary_operation_type_I}. 
\end{proof}

\begin{proposition}
\label{prop:duality_II_complement_diagram}
Let $\bm\Sigma$ be an irregular curve with boundary data of type I AD-$A$ type, with reduced parameter $\mathscr T=(m,k,[Y])$.
Assume that $k=\frac{s}{r}>1$, with $s, r$ coprime, $s>1$; in particular the rank is $N=r$. Let $L$ be the number of columns of $Y$, and let us write $r=\kappa s +\rho$, with $1\leq \rho\leq s-1$. Assume that $Ls>r$, i.e. $L>\kappa$.

Then the irregular curve with boundary data  $(\widetilde{F}^+)^{(L-1)}F\cdot\bm\Sigma
$ is also of type I AD-$A$ type, with reduced parameter
\begin{equation}
\label{eq:duality_complement_diagram}
(\widetilde{F}^+)^{L-1-\kappa} F (\widetilde{F}^-)^{\kappa} \cdot\mathcal T=\left(m, \frac{s}{Ls-r}, [Y^c]\right)
\end{equation}
where the Young diagram $[Y^c]$ is defined from $[Y]$ as follows: if $L$ denotes the number of columns of $[Y]$, and $h_1\geq \dots\geq h_L$ are the heights of the columns of $[Y]$ from left to right, then $[Y^c]$ is the Young diagram with $L$ columns, whose column heights from left to right are given by $(ms-h_L)\geq \dots \geq (ms-h_1)$ (equivalently if $[Y]=[N^{l_N}, \dots, 1^{l_1}]$, then
$[Y^c]=[(ms-1)^{l_1}, \dots, (ms-N)^{l_N}]$.)
\end{proposition}

Notice that the transformation $(\widetilde{F}^+)^{L-1-\kappa} F (\widetilde{F}^-)^\kappa$ amounts to taking $L$ steps given by the blue arrows in Fig. \ref{fig:structure_orbit}, starting from $\mathcal T$.

We can actually give a closed formula for the parameters appearing at the intermediate steps:

\begin{lemma}
\label{thm:intermediate_duality_type_I}
Keeping the notations of Thm. \ref{prop:duality_II_complement_diagram}, we have:
\begin{itemize}
\item For $0\leq l\leq \kappa$,
\[
F (\widetilde{F}^-)^{l} \cdot\mathcal T=\left(m, \frac{s}{(\kappa-l)s+\rho}, [Y_l],[\widetilde Y_l]\right).
\] 
\item Next, after applying $F$:
\[
F (\widetilde{F}^-)^{\kappa} \cdot\mathcal T=\left(m, \frac{s}{s-\rho}, [\widetilde Y_{\kappa+1}], [Y_{\kappa+1}]\right).
\] 
\item Finally, for $\kappa+2\leq l\leq L$,
\[
(\widetilde{F}^+)^{l-1-\kappa} F (\widetilde{F}^-)^{\kappa} \cdot\mathcal T=\left(m, \frac{s}{ls-r}, [\widetilde Y_l], [Y_l]\right).
\] 
\end{itemize}
where, for $l\in \{0, \dots, L\}$, $[Y_l]:=[\tau^l(Y)]$ is the Young diagram with $L-l$ columns, whose columns heights from left to right are $h_{1+l}\geq \dots \geq h_L$, and $[\widetilde Y_l]$ is the Young diagram with $l$ columns, whose column heights from left to right are $(ms-h_{l})\geq \dots \geq (ms-h_{1})$. In particular $[Y_0]=[Y]$, $[\widetilde Y_0]=[\varnothing]$, $[Y_L]=[\varnothing]$, $[\widetilde Y_L]=[Y^c]$. 
\end{lemma}

\begin{proof}
Once again, this follows directly by induction from Prop. \ref{prop:transformation_parameter_elementary_operation_type_I}.
\end{proof}

The process can be summarized as follows: the complete transformation, consisting in taking $L$ blue steps in Fig. \ref{fig:structure_orbit}, changes the Young diagram $[Y]$ to its complement $[Y^c]$. At each intermediate step, one further column of $[Y]$ is removed, and its complement is added to the second Young diagram, building $[Y^c]$ one column at a time. 

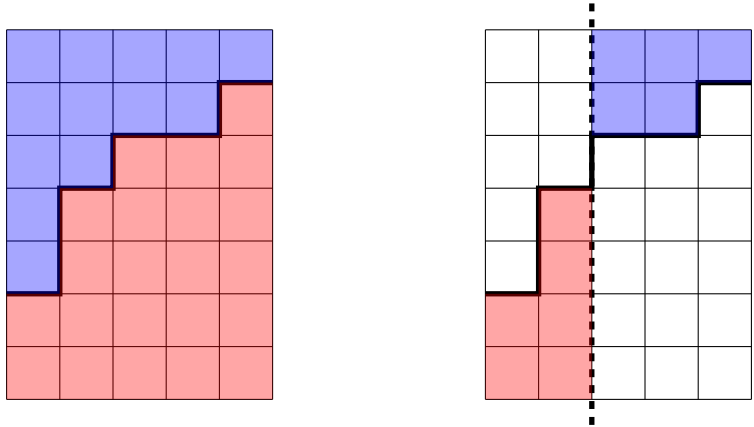
\begin{figure}[h]
\begin{tikzpicture}[scale=0.7]
\begin{scope}
\foreach \i in {0,...,5} \draw  (\i,0)--(\i, -7);
\foreach \j in {0,...,7} \draw  (0,-\j)--(5, -\j);
\draw[line width=0.7mm] (0,-5)--(1,-5)--(1,-3)--(2,-3)--(2,-2)--(4,-2)--(4,-1)--(5,-1);

\fill[blue, opacity=0.35] (0,0)--(0,-5)--(1,-5)--(1,-3)--(2,-3)--(2,-2)--(4,-2)--(4,-1)--(5,-1)--(5,0)--cycle;

\fill[red, opacity=0.35] (0,-7)--(0,-5)--(1,-5)--(1,-3)--(2,-3)--(2,-2)--(4,-2)--(4,-1)--(5,-1)--(5,-7)--cycle;
\end{scope}

\begin{scope}[xshift=9cm]
\foreach \i in {0,...,5} \draw  (\i,0)--(\i, -7);
\foreach \j in {0,...,7} \draw  (0,-\j)--(5, -\j);

\draw[line width=0.7mm] (0,-5)--(1,-5)--(1,-3)--(2,-3)--(2,-2)--(4,-2)--(4,-1)--(5,-1);

\draw[dashed, line width=0.7mm](2, 0.5)--(2,-7.5);

\fill[red, opacity=0.35] (0,-7)--(0,-5)--(1,-5)--(1,-3)--(2,-3)--(2,-7)--cycle;

\fill[blue, opacity=0.35] (2,0)--(2,-2)--(4,-2)--(4,-1)--(5,-1)--(5,0)--cycle;
\end{scope}
\end{tikzpicture}
\caption{The duality transforming the Young diagram $[Y]$ into its complement $[Y^c]$, and the corresponding intermediate steps, drawn for the example $[Y]=[5,3, 2^2, 1]$, $m=1$, $s=7$. The diagram $[Y]$ is represented in blue on the left part and the diagram $[Y^c]$ is represented in red, up to reading the columns from right to left. The right part of the figure shows an intermediate step: at each step one column of the blue diagram is removed, and its complement is added to the red diagram. The intermediate parameter thus features two nontrivial Young diagrams $[Y_l]$ and $[\widetilde Y_l]$.}
\label{fig:duality_complement_intermediate_step}
\end{figure}

\begin{remark}
Notice that the parameters appearing at the intermediate steps all feature two nontrivial Young diagrams, hence they are not of standard type I AD-$A$ type, i.e. they do not directly correspond to a type I Argyres--Douglas theory. On the other hand, all other elements in the orbit $\mathcal O(\mathcal T)$ feature a single nontrivial Young diagram, corresponding to the regular singularity at 0, i.e. they are of standard AD-$A$ type, and correspond to an Argyres--Douglas theory. 

Moreover, all parameters of standard type I in the orbit can be obtained from $\mathcal T$ by a composition of the dualities \eqref{eq:main_duality_add_column} and \eqref{eq:duality_complement_diagram}. 
\end{remark}

\section{General case}
\label{sec:dualities_general_case}

\subsection{Kummer twists}

We now discuss the general case, i.e. we do not assume that the conjugacy classes $\mathcal C_0$ and $\mathcal C_\infty$ are unipotent any more. 

To proceed as in the type I case and use the Fourier transform to change the Young diagrams, it will be necessary to do some shifts of the eigenvalues to reduce to the situation where 1 is an eigenvalue. To perform such shifts, we have to introduce another type of elementary operations on connections, the Kummer twists.

\begin{definition}
Let $\lambda\in \mathbb C$. The Kummer connection $\mathcal K_\lambda$ is the rank one algebraic connection
\[
\left(\mathcal O_{\mathbb A^1\smallsetminus \{0\}}, d-\frac{\lambda}{z}dz\right).
\]
\end{definition}

It has a regular singularity at $0$ and at $\infty$, and the conjugacy classes of its formal monodromies are $\mathcal C_{\cir{0}_0}=\mathcal C_{\cir{0}_\infty}=\{\alpha\}$, where $\alpha=e^{2i\pi\lambda}\in \mathbb C^*$ (with our convention of taking the reverse orientation for the formal monodromy at infinity), in particular up to isomorphism $\mathcal K_\lambda$ only depends on $\alpha$. This immediately implies the following:

\begin{lemma}
Let $(E,\nabla)$ be a connection on a Zariski open subset of $\mathbb P^1$, and $\bm\Sigma$ its irregular curve with boundary data. Then for $\lambda\in \mathbb C$, the irregular curve with boundary data $\bm\Sigma'$ of the tensor product
\[
(E,\nabla)\otimes \mathcal K_\lambda
\]
only depends on $\bm\Sigma$ and $\alpha=e^{2i\pi\lambda}$.
\end{lemma}

\begin{definition}
Keeping the notations of the previous lemma, for  $\alpha\in \mathbb C^*$, we define the Kummer twist of $\bm\Sigma$ as $T_\alpha\cdot \bm\Sigma:=\bm\Sigma'$.
\end{definition}

\begin{lemma}
Let $\bm\Sigma$ be an irregular curve with boundary data of generalized AD-$A$ type, with parameter $\mathcal T=(m, k, \mathcal C_0, \mathcal C_\infty)$. Then $T_\alpha\cdot\bm\Sigma$ is an irregular curve with boundary data of generalized AD-$A$ type, with parameter 
\[
\mathcal T=(m, k, \alpha\;\mathcal C_0, \alpha\;\mathcal C_\infty)
\]
(this uses our convention of taking the reverse orientation for the formal monodromies at infinity.)
\end{lemma}

\begin{proof} This follows directly from the local models \eqref{eq:local_model_connection}, and the relation between the conjugacy classes and the residues \eqref{eq:additive_formal_monodromy}. In some more detail, let $(E,\nabla)$ be an algebraic connection on a Zariski open subset of $\mathbb P^1$ with formal data $\bm\Sigma$. For the regular singularity at $z=0$, in a local trivialization we have $\nabla=d-\frac{\Lambda_0}{z}dz$, with constant residue matrix $\Lambda_0$ such that $\mathcal C_0$ is the conjugacy class of $\exp(2i\pi\Lambda_0)$. For $\lambda\in \mathbb C$ such that $e^{2i\pi\lambda}=\alpha$, tensoring $\nabla$ by $d-\frac{\lambda}{z}dz$ yields 
\[
d-\left(\frac{\Lambda_0}{z}+\frac{\lambda}{z}\mathrm{Id}\right)dz, 
\]
i.e. the residue is shifted by $\lambda\Id$. The new conjugacy class associated to $\cir{0}_0$ is thus the one of $\exp(2i\pi(\Lambda_0+\lambda\Id))$, i.e. $\alpha\;\mathcal C_0$.

Similarly, at $\infty$, using the local coordinate $\zeta:=1/z$, in a local trivialization we have $\nabla=d-A d\zeta$. Using the notations of \eqref{eq:local_model_connection} and \eqref{eq:additive_formal_monodromy}, if $\Lambda_\infty$ is the block in the residue $A_1$ associated to the vanishing exponential factor $q=0$, $\mathcal C_\infty$ is the conjugacy class of $\exp(2i\pi\Lambda_\infty)$. Now, tensoring $\nabla$ by $d-\frac{\lambda}{z}dz=d-\lambda{\zeta}d\zeta$ yields
\[
d-\left(A+\frac{\lambda}{\zeta}\mathrm{Id}\right)d\zeta, 
\]
i.e. the residue $A_1$, and in particular its block $\Lambda_\infty$, are shifted by $\lambda\Id$, hence the new conjugacy class of $\cir{0}_\infty$ is the one of $\exp(2 i\pi(\Lambda_\infty+\lambda\Id))$, i.e. $\alpha\;\mathcal C_\infty$. 
\end{proof}

It will be convenient to introduce a variant of the truncations and extensions:

\begin{definition}
\label{def:truncation_variants}
Let $\mathcal C\subset \mathrm{GL}_N(\mathbb C)$ be a conjugacy class, $\alpha\in \mathbb C^*$ and $A\in \mathcal C$. 
\begin{itemize}
\item The $\alpha$-truncation $\tau_\alpha(\mathcal C)$ of $\mathcal C$ is the conjugacy class of $A_{\vert{\mathrm{Im}(A-\alpha\Id)}}$ in $\mathrm{GL}_{m}(\mathbb C)$, where $m:=\mathrm{rk}(A-\alpha\Id)$.  (this does not depend on the choice of $A$.)
\item Let $h\geq 0$ be an integer. The $\alpha$-extension $\varepsilon_{h,\alpha}(\mathcal C)$ is the conjugacy class $\mathcal C'\subset \mathrm{GL}_{N+h}(\mathbb C)$ such that $\tau_\alpha(\mathcal C')=\mathcal C$ (this is well-defined).
\end{itemize}
\end{definition}

\begin{remark}
The usual truncation $\tau$ considered before corresponds to the particular case of $\tau_\alpha$ with $\alpha=1$. Conversely, $\tau_\alpha$ for general $\alpha$ can be expressed in terms of $\tau$ as follows: for any conjugacy class $\mathcal C$ and $\alpha\in \mathbb C^*$, and $N\geq \mathrm{rk}(\mathcal C)$, we have
\[
\tau_{\alpha}(\mathcal C)=\alpha\;\tau(\alpha^{-1}\mathcal C).
\]
\end{remark}

In terms of the Jordan forms, we have the following explicit expressions:

\begin{lemma} Let 
 $\mathcal C$ be a conjugacy class, and $\alpha\in \mathbb C$. We can  write $\mathcal C=\{(\alpha_1, [Y_{\alpha_1}]), \dots, (\alpha_p, [Y_{\alpha_k}])\}$, with that $\alpha=\alpha_i$ for some $i\in \{1, \dots, p\}$ (with $[Y_{\alpha}]$ possibly empty). Then
 \[
 \tau_\alpha(\mathcal C)=\{(\alpha_1, [Y_{\alpha_1}]), \dots,(\alpha, [\tau(Y_\alpha)]), \dots,(\alpha_p, [Y_{\alpha_k}])\}.
 \]
Conversely, for $h\geq 0$, 
 \[
 \varepsilon_{h,\alpha}(\mathcal C)=\{(\alpha_1, [Y_{\alpha_1}]), \dots,(\alpha, [h,Y_\alpha)]), \dots,(\alpha_p, [Y_{\alpha_k}])\}.
 \]
\end{lemma}

\subsection{Elementary transformations}

We define generalizations of the type I elementary AD-$A$ operations, by conjugating by Kummer twists:

\begin{definition}
The elementary AD-operations are the following operations on irregular curves with boundary data on $\mathbb P^1$:
\begin{itemize}
\item $F_\alpha:=T_{\alpha}FT_{\alpha^{-1}}$;
\item $\widetilde F^+_\alpha:=T_{\alpha}FMT_{\alpha^{-1}}$;
\item $\widetilde F^-_\alpha:=T_{\alpha}MFT_{\alpha^{-1}}$;
\item $T_\alpha$.
\end{itemize}
\end{definition}

In a similar way as in type I, the irregular curves with boundary data are preserved by suitable elementary AD-$A$ operations.

\begin{proposition}
\label{prop:transf_parameter_elementary_transf_twist}
Let $\bm\Sigma$ be an irregular curve with boundary data of generalized  AD-$A$ type, with parameter $\mathcal T=(m,k,\mathcal C_0, \mathcal C_\infty)$, and $\alpha\in \mathbb C^*$. Let $[Y^0_\alpha]$ denote the Young diagram associated to $\alpha$ in the Jordan form of $\mathcal C_0$, and $h:=h_1(Y^0_\alpha)$ the height of its first column. 
\begin{enumerate}
\item If $k\geq 1$, then the irregular curve with boundary data $F_\alpha\cdot \bm\Sigma$ is  of generalized AD-$A$ type, with parameter
\begin{equation}
F_\alpha\cdot \mathcal T:=\left(m,\frac{s}{s-r},\varepsilon_{(ms-h),\alpha}(\mathcal C_{\infty}), \tau_\alpha(\mathcal C_0)\right).
\end{equation}
\item The irregular curve with boundary data $\widetilde F^+_\alpha\cdot \bm\Sigma$ is of generalized  AD-$A$ type, with parameter
\begin{equation}
\widetilde F^+_\alpha\cdot \mathcal T:=\left(m,\frac{s}{s+r},\varepsilon_{(ms-h),\alpha}(\mathcal C_{0}), \tau_\alpha(\mathcal C_\infty)\right) .
\end{equation}
\item If $k<1$, the irregular curve with boundary data $\widetilde F^-_\alpha\cdot \bm\Sigma$ is of  generalized AD-$A$ type, with parameter
\begin{equation}
\widetilde F^-_\alpha\cdot \mathcal T:=\left(m,\frac{s}{r-s},\tau_\alpha(\mathcal C_{0}), \varepsilon_{(ms-h),\alpha}(\mathcal C_\infty)\right).
\end{equation}
\end{enumerate}
\end{proposition}

\begin{proof}
Up to using Lemma \ref{lemma:truncation_in_terms_of_Young_diag}, the proof is the same as the one of Prop. \ref{prop:transformation_parameter_elementary_operation_type_I}.
\end{proof}

\subsection{Orbits}

As in type I, we can study the orbits under compositions of elementary operations.

\begin{definition}

Let $\bm\Sigma$ be an irregular curve with boundary data of generalized AD-$A$ type. We define $\mathcal O(\bm\Sigma)$ of all irregular curves with boundary data of generalized type I AD-$A$ type of the form $O_k\dots O_1\cdot \bm\Sigma$, where $k\geq 1$, and for each $i\in \{1, \dots, k\}$, $O_i$ is an allowed elementary AD-operation on $O_{i-1}\dots O_1\cdot \bm\Sigma$. 
If $\mathcal T$ is a generalized AD-$A$ parameter we define in the same way a set $\mathcal O(\mathcal T)$.
\end{definition}

We call again $\mathcal O(\bm\Sigma)$ (resp. $\mathcal O(\mathcal T)$) the orbit of $\bm\Sigma$ (resp. $\mathcal T$) under AD-$A$ transformations.

\begin{proposition}
Let $\mathcal T=(m, k, \mathcal C_0, \mathcal C_\infty)$ be a generalized AD-$A$ parameter. Let us write $k=\frac{s}{r}$, with $s, r$ coprime, and assume that $s>1$. Let $r=\kappa s+\rho$ , with $\kappa\in \mathbb Z_{\geq 0}$ and $\rho\in \{1, \dots, s-1\}$ be the euclidean division of $r$ by $s$. Then the set of all slopes of the wild Stokes circles of elements in the orbit $\mathcal O(\mathcal T)$ only depends on $k$, and is given by
\begin{equation}
\label{eq:form_slopes_orbit_general_case}
\mathcal O(k):=\left\lbrace \frac{s}{ls\pm \rho} \;\middle \vert\; l\in \mathbb Z_{\geq 0}, ls\pm \rho>0 \right\rbrace.
\end{equation}
\end{proposition}

\begin{proof}
Completely similar to the type I case. 
\end{proof}

\begin{remark}
As far as the slopes are concerned, the situation of Fig. \ref{fig:structure_orbit} remains true in the general case. However, unlike in the type I case, there is no longer a unique element in the orbit corresponding to each vertex of the figure, since we now have a choice of twist $\alpha$ at each step. 
\end{remark}

\subsection{Dualities between standard parameters}

As in type I, in the general case we can obtain more general dualities by composing elementary transformations. In particular, we can recover in this way the third main duality given in \cite{beem2025simplifying}, as follows:

\begin{proposition}
\label{prop:duality_III_gen_case}
Let $\bm\Sigma$ be an irregular curve with boundary data of standard AD-$A$ type, with reduced parameter $\mathscr T=(m, k, [Y])$, parameter $\mathcal T=(m, k, \mathcal C_0, \mathcal C_\infty)$, and rank $N$. Let $\kappa=N-rm$, and $\beta_1, \dots, \beta_\kappa$ be the eigenvalues of $\mathcal C_\infty$. Then $\widetilde{F}^+_{\beta_\kappa}\dots\widetilde{F}^+_{\beta_1}\cdot \bm\Sigma$ has parameter of the form
\[
\widetilde{F}^+_{\beta_\kappa}\dots\widetilde{F}^+_{\beta_1}\cdot \mathcal T=\left(m, \frac{s}{\kappa s+r}, \mathcal C'_0, \{\;\}\right)
\]
where $\{\;\}$ denotes the trivial (rank 0) conjugacy class, and
\begin{equation}
{\mathcal  C}'_0 =\{(1, [Y]),(\beta_1,[ms-1]), \dots,(\beta_\kappa,[ms-1])\}.
\end{equation}
In particular, the nilpotent part of $\mathcal C'_0$ only depends on $\mathscr T$, and is given by the Young diagram
\begin{equation}
[(ms-1)^\kappa, Y].
\end{equation}
\end{proposition}

\begin{proof} Using the second part of Prop.  \ref{prop:transf_parameter_elementary_transf_twist}, it is straightforward to show by induction that for $l\in\{1, \dots, \kappa\}$, we have
\[
\widetilde{F}^+_{\beta_l}\dots\widetilde{F}^+_{\beta_1}\cdot \mathcal T=\left(m, \frac{s}{l s+r}, \widetilde{\mathscr C}'_l,\mathscr C'_l \right)
\]
where for $l\in \{0, \dots, \kappa\}$
\begin{align*}
{\mathscr  C}_l &:=\{(1, [Y]),(\beta_1,[ms-1]), \dots,(\beta_l,[ms-1])\},\\
\widetilde{\mathscr C}_l& :=\{(\beta_{l+1},[1]), \dots,(\beta_\kappa,[1])\}
\end{align*}
(in particular $\mathscr  C_0=\mathcal C_\infty$, $\mathscr  C_\kappa=\{\;\}$, $\widetilde{\mathscr C}_\kappa=\mathcal C'_0$.)
\end{proof}

\begin{remark}
\label{remark:hypermultiplets}
Notice that here $\mathcal C'_0$ is not unipotent, hence the parameter $\widetilde{F}^+_{\beta_\kappa}\dots\widetilde{F}^+_{\beta_1}\cdot \mathcal T$ is not of standard AD-$A$ type, i.e. does not directly correspond to an Argyres--Douglas theory. However, this is consistent with the results of \cite{beem2025simplifying} in the following way: there in §4.4 it is observed that there is a discrepancy between the naive flavour symmetries on both sides, which requires adding some free hypermultiplets. In our framework this corresponds to the fact that the columns of height $ms-1$ in the Young diagram $[(ms-1)^\kappa, Y]$ of $\mathcal C'_0$ do not correspond to the eigenvalue 1.
\end{remark}

\section{Nonabelian Hodge diagrams and 3d mirrors}
\label{sec:diagrams}

In this section, we discuss the relation between the nonabelian Hodge diagrams associated to irregular curves with boundary data of standard AD-$A$ type via the construction of \cite{boalch2020diagrams, doucot2021diagrams}, and the quivers describing the 3d mirrors of the corresponding Argyres--Douglas theories. 

More precisely, we show that if $\bm\Sigma$ is an irregular curve with boundary data of standard AD-$A$ type, the corresponding 3d mirror quiver is a nonabelian Hodge diagram: it is not the diagram $\Gamma(\bm\Sigma)$ associated to $\bm\Sigma$, but the unique one associated to another irregular curve with boundary data $\bm\Sigma'$, of generalized AD-$A$ type, belonging to the orbit of $\bm\Sigma$ under elementary AD operations, having no negative edge-loops and minimal number of vertices. 

\subsection{Nonabelian Hodge diagrams of AD-$A$ type}

Let us briefly recall the construction of nonabelian Hodge diagrams for irregular connections on $\mathbb P^1$. By diagram, we mean the following generalization of a graph:

\begin{definition}[\cite{boalch2020diagrams}]
A \textit{diagram} is a pair $\Gamma=(V,B)$ where $V$ is a finite set (the set of vertices) and $B=(B_{ij})_{i,j\in V}$ is a symmetric square matrix with integer values, 
such that $B_{ii}$ is even for any $i\in V$.
A \textit{dimension vector} $\mathbf d$ for $\Gamma$ is an element of $\mathbb Z_{\geq 0}^{V}$. 
\end{definition}

For $i, j\in V$ such that $i\neq j$, we view the integer $B_{ij}$ as a (possibly negative) number of edges between $i$ and $j$, and for $i\in V$, we view $\frac{B_{ii}}{2}$ as a (possibly negative) number of edge-loops at the vertex $i$.

The main result of \cite{doucot2021diagrams} is the following: if $(E,\nabla)$ is a connection on a Zariski open subset of $\mathbb P^1$, there is an explicit way to associate to it a nonabelian Hodge diagram $\Gamma(E, \nabla)$, in such a way that the diagram is invariant under the Fourier transform, i.e.
\[
\Gamma(E,\nabla)=\Gamma(F\cdot (E,\nabla)).
\]
The diagram only depends on the formal data $(\bm\Theta, \bm{\mathcal C})$ of $(E,\nabla)$, and is expressed most conveniently in terms of its modified formal data $(\bm{\breve\Theta}, \bm{\breve{\mathcal C}})$, so we write $\Gamma(E,\nabla)=\Gamma(\bm{\breve\Theta}, \bm{\breve{\mathcal C}})$. 

Furthermore, writing $\Gamma(E,\nabla)=(V,B)$,  for any choice of minimal marking (see below) of $\bm{\breve{\mathcal C}}$, we obtain a dimension vector $\mathbf d\in \mathbb Z^V$ for $\Gamma(E,\nabla)$, and the dimension of the (symplectic) wild character variety $\mathcal M_B(E,\nabla)$ is given by the formula
\[
2-(\mathbf d, \mathbf d),
\]
where $(\cdot, \cdot)$ is the bilinear form on $\mathbb Z^V$ defined by the Cartan matrix $C:=2\Id-B$.\\

Let us briefly sketch the construction: the diagram $\Gamma(\bm{\breve\Theta}, \bm{\breve{\mathcal C}})$ consists of a \emph{core} $\Gamma_c(\bm{\breve\Theta})$, only depending on the global modified irregular classs $\bm{\breve\Theta}$, to which are glued  \emph{legs} encoding the conjugacy classes in $\bm{\breve{\mathcal C}}$ (they corresponds to the ``tails'' in the physics terminology).

The core diagram $\Gamma_c(\bm{\breve\Theta})$ is constructed as follows. For any two Stokes circles $\cir{q}, \cir{q'}\in \pi_0(\mathcal I)$ (possibly identical), we define an integer $B_{\cir{q}\cir{q'}}=B_{\cir{q'}\cir{q}}\in \mathbb Z$, such that $B_{\cir{q}\cir{q}}$ is even (we will not recall the detailed definition here, let us just mention that the motivation is that when $\cir{q}, \cir{q'}$ are both at infinity, $B_{\cir{q}\cir{q'}}$ is closely related to the number of nontrivial entries between $\cir{q}$ and $\cir{q'}$ in the Stokes matrices, in the explicit presentations of wild character varieties.)

Now, writing
\[
\bm{\breve\Theta}=\sum_{i=1}^p m_i\cir{q_i},
\]
with $m_i\geq 1$ for $i\in\{1, \dots, p\}$, the core diagram $\Gamma_c(\bm\Theta)$ is the diagram whose vertices are identified with the Stokes circles $\cir{q_i}$, for $i\in \{1, \dots, p\}$ and whose edge/loop multiplicities are defined as $B_{\cir{q_i}, \cir{q_j}}$, for $i,j\in \{1, \dots, p\}$.

Next, writing $\bm{\breve{\mathcal C}}=(\breve{\mathcal C}_1, \dots, \breve{\mathcal C}_p)$, with $\breve{\mathcal C}_i\subset\mathrm{GL}_{m_i}(\mathbb C)$, for $i\in \{1, \dots, p\}$, the full diagram is obtained by gluing to the vertex $\cir{q_i}$ a leg encoding the conjugacy class $\breve{\mathcal C}_{\cir{q_i}}$, for all $i\in \{1, \dots, p\}$. To define these legs, let us first recall the notion of marking:

\begin{definition}
Let $\mathcal C\subset\GL_N(\mathbb C)$ be a conjugacy class. 
A marking of $\mathcal C$ is a tuple $\xi=(\xi_1, \dots, \xi_k)$, where $\xi_i\in \mathbb C^*$ for $i\in \{1, \dots, k\}$, such that
\[
(A-\xi_1)\dots (A-\xi_k)=0
\]
for $A\in \mathcal C$.
\end{definition}

We say that a marking $\xi=(x_1, \dots, \xi_w)$ is \emph{minimal} if its number of elements $k$ is minimal. In that case the polynomial $(X-\xi_1)\dots(X-\xi_w)$ is the minimal polynomial of $\mathcal C$, i.e. a minimal marking amounts to a choice of ordering of its eigenvalues, counted with multiplicity.

\begin{definition}
Let $\mathcal C\subset\GL_N(\mathbb C)$ be a conjugacy class, and $\xi=(x_1, \dots, \xi_w)$ a marking of $\mathcal C$. The \emph{leg} associated to $(\mathcal C, \xi)$ is the pair $(\mathbb L, \mathbf d)$, where $\mathbb L$ is the Dynkin diagram of type $A_w$ (see the figure below),
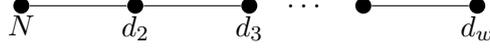
\vspace{0.3cm}
\begin{center}
\begin{tikzpicture}
\tikzstyle{vertex}=[circle,fill=black,minimum size=6pt,inner sep=0pt]
\foreach \name/\x in {2/2,3/3,w/5}
    {\node[vertex] (\name) at (1.5*\x,0){};
    \draw (1.5*\x,0) node[below] {$d_{\name}$};}
\node[vertex] (1) at (1.5,0){};
\draw (1.5,0) node[below]{$N$};
\draw (1)--(2)--(3);
\draw (1.5*3.5,0) node {$\dots$};
\node[circle,fill=black,minimum size=6pt,inner sep=0pt] (4) at (1.5*4,0){};
\draw (4)--(w);
\end{tikzpicture}
\end{center}
and $\mathbf d=(d_1, \dots, d_w)\in \mathbb Z^w$ is the dimension vector for $\mathbb L$ given by
\begin{equation}
d_i=\left\lbrace
\begin{array}{ll}
N & \text{ for } i=1,\\
\mathrm{rk}\left((A-\xi_1)\dots(A-\xi_{i-1})\right) & \text{ for } i\in \{2, \dots, w\}, \text{ where } A\in\mathcal C.
\end{array}
\right.
\end{equation}

\end{definition}

In particular, if we require the marking to be minimal, then $\mathbb L$ is fully determined by $\mathcal C$. If moreover $\mathcal C$ has a single eigenvalue (in particular if $\mathcal C$ is unipotent), then the dimension vector $\mathbf d$ is also fully determined by $\mathcal C$. 

Now, choose a minimal marking $\bm\xi$ of $\breve{\bm{\mathcal C}}$, i.e. the datum of a minimal marking $\xi_i$ of $\breve{\mathcal C}_{\cir{q_i}}$ for each $i\in \{1, \dots, p\}$.
For $i\in \{1, \dots, p\}$, let $(\mathbb L_i, \mathbf d_i)$ be the leg associated to $(\breve{\mathcal C}_{\cir{q_i}}, \xi_i)$.
We obtain the full diagram $\Gamma(\bm{\breve\Theta}, \bm{\breve{\mathcal C}})$, together with a dimension vector $\mathbf d$ for it, by fusing the extremity with dimension $m_i$ of the leg $(\mathbb L_i, \mathbf d_i)$ with the core vertex $\cir{q_i}$, for each $i\in \{1, \dots, p\}$. Notice that the diagram itself does not depend on the choice of marking. \\

To draw the diagrams in the case of irregular curves with boundary data of AD-$A$ type that we are interested in here, we just need the expressions for the edge/loop multiplicities $B_{\cir{q}, \cir{q'}}$ for the types of Stokes circles appearing in irregular curves with boundary data of generalized AD-$A$ type.

\begin{lemma}
Consider a global modified irregular class of AD-$A$, type, i.e.
\[
\bm{\breve\Theta}=\cir{q_1}_\infty+\dots \cir{q_m}_\infty+ \breve n_{\cir{0}_0} \cir{0}_0+ n_{\cir{0}_\infty}\cir{0}_\infty.
\]
with $\cir{q_i}$ of slope $k=\frac{s}{r}$ with $s,r$ coprime for $i\in \{1, \dots, m\}$. We have:
\begin{itemize}
\item $B_{{\cir{q_i}}_\infty {\cir{q_i}}_\infty}=(r-1)(s-r-1)$.
\item $B_{{\cir{0}}_\infty {\cir{0}}_\infty}=0$.
\item $B_{{\cir{0}}_0 {\cir{0}}_0}=0$.
\item $B_{{\cir{q_i}}_\infty {\cir{q_j}}_\infty}=r(s-r)$ if $i\neq j$.
\item $B_{{\cir{q_i}}_\infty {\cir{0}}_\infty}=(s-r)$
\item $B_{{\cir{q_i}}_\infty {\cir{0}}_0}=r$
\item $B_{{\cir{0}}_\infty {\cir{0}}_0}=1$.
\end{itemize}
In other words, the core nonabelian Hodge diagram $\Gamma_c(\bm{\breve\Theta})$ is given by Fig. \ref{fig:core_diagram_AD_type}.
\end{lemma}

\begin{figure}[H]
\begin{tikzpicture}
\tikzstyle{vertex}=[circle,fill=black,minimum size=6pt,inner sep=0pt]
\node[vertex] (T0) at (0,0){};
\node[vertex] (Tinf) at (4,0){};
\node[vertex] (W1) at (0,2){};
\node[vertex] (W2) at (2,3){};
\node[vertex] (W3) at (4,2){};
\draw (T0) --node[midway, below]{$1$}(Tinf);
\draw (T0) --node[midway, left]{$r$}(W1);
\draw (T0) --node[near start, right]{$r$}(W2);
\draw (T0) --node[near start, below right]{$r$}(W3);
\draw (Tinf) --node[midway, right]{$s-r$}(W3);
\draw (Tinf) --node[near start, above left]{$s-r$}(W2);
\draw (Tinf) --node[near start, left]{$s-r$}(W1);
\draw (W1) --node[near end, left]{$r(s-r)$}(W2);
\draw (W1) --node[midway, above]{$r(s-r)$}(W3);
\draw (W2) --node[near start, right]{$r(s-r)$}(W3);
\draw (W1) to[out=90, in=45] (-0.5,2.5) to[out=225,in=180] (W1);
\draw (W2) to[out=45, in=0] (2,3.7) to[out=180,in=135] (W2);
\draw (W3) to[out=0, in=-45] (4.5,2.5) to[out=135,in=90] (W3);
\draw (-1.4, 2.8) node {$\frac{(r-1)(s-r-1)}{2}$};
\draw (5.4, 2.8) node {$\frac{(r-1)(s-r-1)}{2}$};
\draw (2, 4.2) node {$\frac{(r-1)(s-r-1)}{2}$};
\draw (-0.5, -0.5) node {$\cir{0}_0$};
\draw (4.5, -0.5) node {$\cir{0}_\infty$};
\draw (-1, 2) node {$\cir{q_1}_\infty$};
\draw (5, 2) node {$\cir{q_m}_\infty$};
\draw (2.7, 3.2) node {$\cir{q_i}_\infty$};
\end{tikzpicture}
\caption{The core nonabelian Hodge diagram of an irregular curve with boundary data of AD-$A$ type with paremeter $\mathcal T=(m, \frac{s}{r}, \mathcal C_0, \mathcal C_\infty)$ with $s,r$ coprime (drawn for $m=3$). The core vertices are in one-to-one correspondence with the Stokes circles $\cir{q_1}_\infty, \dots, \cir{q_1}_\infty$, $\cir{0}_0$, $\cir{0}_\infty$. The numbers at the middle of the edges/loops indicate the multiplicities (when no loops are drawn the multiplicity is zero). The full diagram is obtained by gluing to the vertex $\cir{0}_0$ a leg encoding $\tau(\mathcal C_0)$, and to $\cir{0}_\infty$ a leg encoding $\mathcal C_\infty$.}
\label{fig:core_diagram_AD_type}
\end{figure}

\begin{proof}
This just follows from the formulas for the edge/loop multiplicities in \cite{doucot2021diagrams}.
\end{proof}

Notice that the core diagram has no negative edges/loops if and only if $k\geq 1$.\\

\begin{lemma}
If $\mathcal T=(m, \frac{s}{r}, \mathcal C_0, \mathcal C_\infty)$ is an AD-$A$ parameter, and $\bm\Sigma$ is an irregular curve with boundary data with parameter $\mathcal T$, then the nonabelian Hodge diagram $\Gamma(\bm\Sigma)$ does not depend on $\bm\Sigma$ (so we can write $\Gamma(\bm\Sigma)=\Gamma(\mathcal T)$), and is constructed as follows:
\begin{enumerate}
\item take the core diagram of Fig. \ref{fig:core_diagram_AD_type};
\item glue to the vertex $\cir{0}_0$
the leg $\mathbb L_0$ associated to $\tau(\mathcal C_0)$;
\item glue to the vertex $\cir{0}_\infty$ the leg $\mathbb L_\infty$ associated to $\mathcal C_\infty$.
\end{enumerate}
In particular, if the conjugacy class $\tau(\mathcal C_0)$ (resp. $\mathcal C_\infty)$) has rank zero, the vertex $\cir{0}_0$ (resp. $\cir{0}_\infty$) has multiplicity zero, so we remove it.

Furthermore, a choice of minimal markings $\xi_0, \xi_\infty$ of $\tau(\mathcal C_0)$, $\mathcal C_\infty$ respectively determines a dimension vector $\mathbf d$ for $\Gamma(\mathcal T)$, taking the dimensions associated to the vertices $\cir{q_i}_\infty$ for $i\in \{1, \dots, m\}$ to be 1. 
\end{lemma}

In particular, in the type I case, the dimension vector $\mathbf d$ is fully determined by $\mathcal T$.

\begin{proof}
This follows immediately from the definition of nonabelian Hodge graphs. Indeed, denoting by $(\bm{\breve\Theta}, \bm{\breve{\mathcal C}})$ the modified formal data corresponding to $\bm\Sigma$, in $\breve{\bm\Theta}$ the irregular Stokes circles $\cir{q_i}_\infty$ for $i\in \{1, \dots, m\}$ have multiplicity 1, so the corresponding vertices have associated dimension 1. Furthermore, in $\bm{\breve{\mathcal C}}$ we have $\bm{\breve{\mathcal C}}_{\cir{0}_0}=\tau(\mathcal C_0)$, so the leg $\mathbb L_0$ has to be the one of the truncation $\tau(\mathcal C_0)$. 
\end{proof}

We say that a diagram $\Gamma=(B,V)$ is of AD-$A$ type if it is of the form $\Gamma(\mathcal T)$ for some AD-$A$ parameter $\mathcal T$. Clearly, given a diagram $\Gamma$, one can  determine whether it is of AD-$A$ type, and when so explicitly determine a type I AD-$A$ parameter $\mathcal T$ such that $\Gamma=\Gamma(T)$ ($\mathcal T$ is unique if $s<r$, and if $r<s$ it is unique up to the simultaneous exchange   $s\leftrightarrow s-r$ and $\mathbb L_0\leftrightarrow \mathbb L_\infty$).

\subsection{Nonnegative AD-$A$ nonabelian Hodge diagrams}

\subsubsection{Type I} Let us first discuss the type I case.

\begin{proposition}
\label{prop:nonnegative_diagram_type_I}
Let $\mathcal T=(m,k, [Y^0], [Y^\infty])$ be a generalized type I AD-$A$ parameter, such $k=\frac{s}{r}$ with $s,r$ coprime and $s>1$. Then there is a unique nonabelian Hodge  diagram of the form $\Gamma(\mathcal T')$, with $\mathcal T' \in \mathcal O(\mathcal T)$ with no negative edge/loops that we denote by $\Gamma_+(\mathcal T)$. Explicitly, we have
\[
\Gamma_+(\mathcal T)=
\Gamma((\widetilde{F}^-)^\kappa\cdot \mathcal T),
\]
where $\kappa:=\lfloor k^{-1}\rfloor$.
\end{proposition}

\begin{proof}
The integer $\kappa:=\lfloor k^{-1}\rfloor$ is the remainder in the euclidean division of $r$ by $s$, i.e. we have $r=\kappa s+\rho$, with $1\leq \rho\leq s-1$. The result follows directly from the structure of the orbit given by Fig. \ref{fig:structure_orbit}: in any case, there are at most two elements $\mathcal T'\in O(\mathcal T)$ with slope $k'\geq 1$, one of them being given by $(\widetilde{F}^-)^\kappa\cdot \mathcal T)$, and when they are exactly two such elements, they are exchanged by $F$. The conclusion follows by invariance of the diagram under Fourier transform.
\end{proof}

We can determine the nonnegative diagram explicitly:

\begin{proposition}
\label{prop:explicit_nonnegative_diagram_type_I}
Let $\mathcal T=(m,k,[Y^0], [Y^\infty])$ be a generalized type I AD-$A$ parameter. Let us write $k=\frac{s}{r}$, and $r=\kappa s+\rho$ with $\kappa\in \mathbb Z$, $\rho\in \{1, \dots, s-1\}$. Then we have $\Gamma_+(\mathscr T)=\Gamma(\mathcal T')$ with
\[
\mathcal T'=\left(m, \frac{s}{\rho}, [Y^0_\kappa], [\widetilde Y^0_\kappa, Y^\infty]\right),
\]
where the Young diagrams $[Y^0_\kappa]$, $[\widetilde Y^0_\kappa, ]$ are defined as in Thm. \ref{thm:intermediate_duality_type_I}. In particular if $k>1$, we have $\Gamma_+(\mathcal T)=\Gamma(\mathcal T)$.
\end{proposition}

\begin{proof}
This is completely similar to Thm. \ref{thm:intermediate_duality_type_I}, using the third part of Prop. \ref{prop:transformation_parameter_elementary_operation_type_I}.
\end{proof}

\begin{remark}
\label{remark:case_irregularity_1}
In the case $s=1$, i.e. $k=\frac{1}{r}$, then for the parameter $\mathcal T'=\Gamma((\widetilde{F}^-)^{r-1}\cdot \mathcal T)$ the wild Stokes circles are of slope 1, i.e. of the form $\cir{q_i}_\infty=\cir{a_i z}_\infty$, with $a_i\in \mathbb C^*$ for $i\in \{1, \dots, m\}$. Now, if $\bm\Sigma'$ is an irregular curve with boundary data with parameter $\mathcal T'$, then $F\cdot \bm\Sigma'$ is no longer of generalized AD-$A$ type. Indeed by the stationary phase formula we have $F\cdot \cir{a_i z}_\infty=\cir{0}_{a_i}$, so $F\cdot \bm\Sigma'$ has $m+2$ regular singularities, at the points $0,\infty, a_1, \dots, a_m$. In particular, $\Gamma_+(\mathcal T):=\Gamma(\bm\Sigma')=\Gamma(F\cdot\bm\Sigma')$ is a star-shaped graph. 
\end{remark}

We notice that the nonnegative diagrams correspond to the 3d mirror quivers of the corresponding Argyres--Douglas theories given in the physics literature.

\begin{proposition} 
\label{prop:3d_mirrors_are_nonneg_diagrams_type_I}
Let $\mathscr T=(m,k, [Y])$ be a reduced type I AD parameter. Then in all cases considered in \cite{beem2025simplifying}, the quiver describing the 3d mirror of the corresponding Argyres--Douglas theory is the nonnegative nonabelian Hodge graph $\Gamma_+(\mathscr T)$, with its uniquely defined dimension vector. 
\end{proposition}

\begin{proof}
In the case $k>1$ and $[Y]=[1^N]$, from Prop. \ref{prop:explicit_nonnegative_diagram_type_I} and Fig. \ref{fig:core_diagram_AD_type}, we immediately obtain the 3d mirror described on Fig. 2 in \cite{beem2025simplifying} (since to get the leg associated to $\cir{0}_0$, we have to pass to the modified global irregular class, i.e. to $[\tau(Y)]=[1^{N-1}]$).
For $k<1$ and $[Y]=[1^N]$, similarly we obtain the type I specialization of Fig. 3 in \cite{beem2025simplifying} (i.e. with $N-nm=0$ is the notations of \emph{loc.~cit.}), using that we have the following dictionary between notations

\begin{center}
\vspace{0.3cm}
\begin{tabular}{|c|c|}
\hline
Irreg. curve  & AD theory\\
\hline
$\kappa$ & $\mu$\\
\hline
$\rho$ & $\nu$\\
\hline
$s$ & $q$\\
\hline
$ms$ & $p$\\
\hline
\end{tabular}
\end{center}
For more general Young diagrams, $\Gamma_+(\mathscr T)$ matches with the 3d mirrors found by quiver subtraction. For instance for instance for the Argyres--Douglas theory $D_3(\mathfrak{sl}_7, [2, 1^5])$ considered in Example 5 of \cite{beem2025simplifying}, the corresponding reduced parameter is $\mathscr T=(1, \frac{3}{7},[2, 1^5])$, from Prop. \ref{prop:explicit_nonnegative_diagram_type_I} the parameter corresponding to the minimal diagram is $\mathcal T'=(1, 3, [1^3], [2,1])$, hence $\Gamma_+(\mathscr T)=\Gamma(\mathcal T')$ is the 3d mirror shown in Eq. (4.9) of \emph{loc.~cit.}
\end{proof}

\subsubsection{General case}

Let us now discuss the general case. Here since there are several possible choices of twist $\alpha$, in general there is no unique parameter with nonnegative nonabelian Hodge diagram. However, in the standard case corresponding to Argyres--Douglas theories, there is a unique preferred choice leading to a nonnegative diagram with minimal number of vertices.

\begin{proposition}
\label{prop:nonnegative_diagram_general_case}
Let $\mathcal T=(m,k, \mathcal C_0, \mathcal C_\infty)$ be a generalized  AD-$A$ parameter, such $k=\frac{s}{r}$ with $s,r$ coprime and $s>1$. Then there is a unique nonabelian Hodge  diagram of the form $\Gamma(\mathcal T')$, with $\mathcal T' \in \mathcal O(\mathcal T)$, such that $\Gamma(\mathcal T')$ has no negative edges/loops, and such that the number of vertices of $\Gamma(\mathcal T')$ is minimal. We denote it by $\Gamma_+(\mathcal T)$. Explicitly, we have
\[
\Gamma_+(\mathcal T)=
\Gamma(T_{\alpha_{\kappa+1}}\widetilde{F}_{\alpha_\kappa}^-\dots \widetilde{F}_{\alpha_1}^- \cdot \mathcal T),
\]
where $\kappa:=\lfloor k^{-1}\rfloor$, and  any $\alpha_1, \dots,\alpha_\kappa$ which are chosen to be eigenvalues of $\mathcal C_0$ `as much as possible', that is if $(\beta_1, \dots, \beta_L)$ is the list of eigenvalues of $\mathcal C_\infty$ counted with multiplicities, i.e. such that the number of columns in the Young diagram encoding the nilpotent part of $(A-h_1\Id)\dots (A-\alpha_\kappa\Id)$ is minimal. 
\end{proposition}

\begin{proof}

The reasoning is similar to Prop. \ref{prop:nonnegative_diagram_type_I}. The main point is that requiring the number of vertices of $\Gamma(\mathcal T')$ is minimal necessitates taking the $\alpha_i$ to be eigenvalues of $\mathcal C_0$. This follows directly from the following observation: by the third part Prop. \ref{prop:transf_parameter_elementary_transf_twist}, applying $(\widetilde{F}^-_\alpha)$ transforms $\mathcal C_0$ into its truncation $\tau_\alpha(\mathcal C_0)$, and $\mathcal C_\infty$ into $\varepsilon_{(ms-h),\alpha}(\mathcal C_\infty)$. At the level of the corresponding legs $\mathbb L_0$ and $\mathbb L_\infty$, this has the following effect: it always adds one vertex to $\mathbb L_\infty$, and for $\mathbb L_0$, it removes a vertex if $\alpha$ is an eigenvalue of $\mathcal C_0$, and does not change $\mathbb L_0$ otherwise (in particular this implies that $\Gamma(\mathcal T')$ does not depend on the choice of eigenvalue at each step). Similarly, to have a diagram with a minimal number of vertices, the final twist $\widetilde{F}_{h_\kappa}$ is required, so that, when constructing the diagram, passing to the modified conjugacy class $\breve{\mathcal C}_0$ removes one further vertex from  the leg $\mathbb L_0$. 
\end{proof}

\begin{remark}
If $\mathcal C_0$ is unipotent (so in particular in the standard case), there is a unique choice of $\alpha_1, \dots, \alpha_{\kappa+1}$, so the dimension vector of $\Gamma_+$ is also uniquely determined.
\end{remark}

This allows us to obtain an explicit formula for the minimal diagram in the standard case:

\begin{proposition}
\label{prop:explicit_nonnegative_diagram_general_case}
Let $\mathscr T=(m,k,[Y])$ be a standard AD-$A$ parameter. Let us write $k=\frac{s}{r}$, and $r=\kappa s+\rho$ with $\kappa\in \mathbb Z$, $\rho\in \{1, \dots, s-1\}$. Then we have $\Gamma_+(\mathscr T)=\Gamma(\widetilde{\mathcal T})$ with
\[
\widetilde{\mathcal T}=\left(m, \frac{s}{\rho}, [Y_\kappa], [\widetilde Y_\kappa, 1^{N-mr}]\right),
\]
where $N=\rm{rk}(Y)$, and the Young diagrams $[Y_\kappa]$, $[\widetilde Y_\kappa]$ are defined as in Lemma \ref{thm:intermediate_duality_type_I}.
\end{proposition}

\begin{proof}
The reduced parameter $\mathscr T=(m,k,[Y])$ corresponds to a non-reduced parameter $\mathscr T=(m,k,\mathcal C_0, \mathcal C_\infty)$, with $\mathcal C_0=\{(1, [Y])\}$, and $\mathcal C_\infty$ of rank $N-mr$ and regular semisimple, so with nilpotent part given by the Young diagram $[1^{N-mr}]$. Then, by Prop. \ref{prop:nonnegative_diagram_general_case}, we have $\Gamma_+(\mathcal T)=\Gamma(\mathcal T')$ with $\mathcal T'=(\widetilde{F}^-)^\kappa$. Then, by induction using the third part of Prop. \ref{prop:transformation_parameter_elementary_operation_type_I}, we obtain that the Young diagrams of the nilpotent parts of the conjugacy classes of $\mathcal T'$ are $[Y_\kappa]$ at $0$ and $[\widetilde Y_\kappa, 1^{N-mr}]$ at $\infty$.
\end{proof}

As in type I, in the general case, the 3d mirror quivers for type $A$ Argyres--Douglas correspond to the minimal diagrams:

\begin{proposition} 
\label{prop:3d_mirrors_are_nonneg_diagrams_gen_case}
Let $\mathscr T=(m,k, [Y])$ be a reduced AD parameter. Then in all cases considered in \cite{beem2025simplifying}, the quiver describing the 3d mirror of the corresponding Argyres--Douglas theory is the minimal positive nonabelian Hodge graph $\Gamma_+(\mathscr T)$, with its uniquely defined dimension vector. 
\end{proposition}

\begin{proof}
In the case $k>1$ and $[Y]=[1^N]$, from Prop. \ref{prop:explicit_nonnegative_diagram_general_case} and Fig. \ref{fig:core_diagram_AD_type}, we immediately obtain the 3d mirror described on Fig. 1 in \emph{loc.~cit.}
For $k<1$ and $[Y]=[1^N]$, similarly we obtain the quiver of Fig. 3 in \emph{loc. cit}. Since the minimal quiver is invariant under elementary AD-$A$ operations, this also holds true for the cases with more general punctures discussed in §4.3 of \emph{loc.~cit.} which are dual to a case with a regular puncture.
\end{proof}

\end{document}